\documentclass[submission, copyright, creativecommons]{eptcs}
\usepackage{algorithm}
\usepackage{algorithmicx} \usepackage[noend]{algpseudocode}
\usepackage{aliascnt}
\usepackage{amsmath}
\usepackage{amssymb}
\usepackage{amsthm}
\usepackage{underscore}
\usepackage{enumitem}
\usepackage[utf8]{inputenc}
\usepackage{mathtools}
\usepackage{microtype}
\usepackage{multicol}
\usepackage{stmaryrd}
\usepackage{thmtools}
\usepackage{tikz}
\usepackage{todonotes}
\usepackage{xifthen}

\usepackage[capitalise,nosort]{cleveref}

\usetikzlibrary{arrows, automata}

\DeclareUnicodeCharacter{00A0}{~}
\DeclareUnicodeCharacter{00AC}{\neg}
\DeclareUnicodeCharacter{00B1}{\pm}
\DeclareUnicodeCharacter{00D7}{\times}
\DeclareUnicodeCharacter{00F7}{\div}
\DeclareUnicodeCharacter{2020}{\dag}
\DeclareUnicodeCharacter{2021}{\ddag}
\DeclareUnicodeCharacter{0391}{A}
\DeclareUnicodeCharacter{0392}{B}
\DeclareUnicodeCharacter{0393}{\varGamma}
\DeclareUnicodeCharacter{0394}{\varDelta}
\DeclareUnicodeCharacter{0395}{E}
\DeclareUnicodeCharacter{0396}{Z}
\DeclareUnicodeCharacter{0397}{H}
\DeclareUnicodeCharacter{0398}{\varTheta}
\DeclareUnicodeCharacter{0399}{I}
\DeclareUnicodeCharacter{039A}{K}
\DeclareUnicodeCharacter{039B}{\varLambda}
\DeclareUnicodeCharacter{039C}{M}
\DeclareUnicodeCharacter{039D}{N}
\DeclareUnicodeCharacter{039E}{\varXi}
\DeclareUnicodeCharacter{039F}{O}
\DeclareUnicodeCharacter{03A0}{\varPi}
\DeclareUnicodeCharacter{03A1}{P}
\DeclareUnicodeCharacter{03A3}{\varSigma}
\DeclareUnicodeCharacter{03A4}{T}
\DeclareUnicodeCharacter{03A5}{\varUpsilon}
\DeclareUnicodeCharacter{03A6}{\varPhi}
\DeclareUnicodeCharacter{03A7}{X}
\DeclareUnicodeCharacter{03A8}{\varPsi}
\DeclareUnicodeCharacter{03A9}{\varOmega}
\DeclareUnicodeCharacter{03B1}{\alpha}
\DeclareUnicodeCharacter{03B2}{\beta}
\DeclareUnicodeCharacter{03B3}{\gamma}
\DeclareUnicodeCharacter{03B4}{\delta}
\DeclareUnicodeCharacter{03B5}{\varepsilon}
\DeclareUnicodeCharacter{03B6}{\zeta}
\DeclareUnicodeCharacter{03B7}{\eta}
\DeclareUnicodeCharacter{03B8}{\theta}
\DeclareUnicodeCharacter{03B9}{\iota}
\DeclareUnicodeCharacter{03BA}{\kappa}
\DeclareUnicodeCharacter{03BB}{\lambda}
\DeclareUnicodeCharacter{03BC}{\mu}
\DeclareUnicodeCharacter{03BD}{\nu}
\DeclareUnicodeCharacter{03BE}{\xi}
\DeclareUnicodeCharacter{03BF}{o}
\DeclareUnicodeCharacter{03C0}{\pi}
\DeclareUnicodeCharacter{03C1}{\rho}
\DeclareUnicodeCharacter{03C3}{\sigma}
\DeclareUnicodeCharacter{03C2}{\varsigma}
\DeclareUnicodeCharacter{03C4}{\tau}
\DeclareUnicodeCharacter{03C5}{\upsilon}
\DeclareUnicodeCharacter{03C6}{\varphi}
\DeclareUnicodeCharacter{03C7}{\chi}
\DeclareUnicodeCharacter{03C8}{\psi}
\DeclareUnicodeCharacter{03C9}{\omega}
\DeclareUnicodeCharacter{03D1}{\vartheta}
\DeclareUnicodeCharacter{03D5}{\phi}
\DeclareUnicodeCharacter{03D6}{\varpi}
\DeclareUnicodeCharacter{03F1}{\varrho}
\DeclareUnicodeCharacter{03F4}{\Theta}
\DeclareUnicodeCharacter{03F5}{\epsilon}
\DeclareUnicodeCharacter{00C1}{\'{A}}
\DeclareUnicodeCharacter{0106}{\'{C}}
\DeclareUnicodeCharacter{00C9}{\'{E}}
\DeclareUnicodeCharacter{01F4}{\'{G}}
\DeclareUnicodeCharacter{00CD}{\'{I}}
\DeclareUnicodeCharacter{1E30}{\'{K}}
\DeclareUnicodeCharacter{0139}{\'{L}}
\DeclareUnicodeCharacter{1E3E}{\'{M}}
\DeclareUnicodeCharacter{0143}{\'{N}}
\DeclareUnicodeCharacter{00D3}{\'{O}}
\DeclareUnicodeCharacter{1E54}{\'{P}}
\DeclareUnicodeCharacter{0154}{\'{R}}
\DeclareUnicodeCharacter{015A}{\'{S}}
\DeclareUnicodeCharacter{1E82}{\'{W}}
\DeclareUnicodeCharacter{00DA}{\'{U}}
\DeclareUnicodeCharacter{00DD}{\'{Y}}
\DeclareUnicodeCharacter{0179}{\'{Z}}
\DeclareUnicodeCharacter{00E1}{\'{a}}
\DeclareUnicodeCharacter{0107}{\'{c}}
\DeclareUnicodeCharacter{00E9}{\'{e}}
\DeclareUnicodeCharacter{01F5}{\'{g}}
\DeclareUnicodeCharacter{00ED}{\'{i}}
\DeclareUnicodeCharacter{1E31}{\'{k}}
\DeclareUnicodeCharacter{013A}{\'{l}}
\DeclareUnicodeCharacter{1E3F}{\'{m}}
\DeclareUnicodeCharacter{0144}{\'{n}}
\DeclareUnicodeCharacter{00F3}{\'{o}}
\DeclareUnicodeCharacter{1E55}{\'{p}}
\DeclareUnicodeCharacter{0155}{\'{r}}
\DeclareUnicodeCharacter{015B}{\'{s}}
\DeclareUnicodeCharacter{1E83}{\'{w}}
\DeclareUnicodeCharacter{00FA}{\'{u}}
\DeclareUnicodeCharacter{00FD}{\'{y}}
\DeclareUnicodeCharacter{017A}{\'{z}}
\DeclareUnicodeCharacter{00C0}{\`{A}}
\DeclareUnicodeCharacter{00C8}{\`{E}}
\DeclareUnicodeCharacter{00CC}{\`{I}}
\DeclareUnicodeCharacter{00D2}{\`{O}}
\DeclareUnicodeCharacter{00D9}{\`{U}}
\DeclareUnicodeCharacter{00E0}{\`{a}}
\DeclareUnicodeCharacter{00E8}{\`{e}}
\DeclareUnicodeCharacter{00EC}{\`{i}}
\DeclareUnicodeCharacter{00F2}{\`{o}}
\DeclareUnicodeCharacter{00F9}{\`{u}}
\DeclareUnicodeCharacter{2013}{--}
\DeclareUnicodeCharacter{2014}{---}
\DeclareUnicodeCharacter{201C}{``}
\DeclareUnicodeCharacter{201D}{''}
\DeclareUnicodeCharacter{201E}{,,}
\DeclareUnicodeCharacter{2026}{\ldots}
\DeclareUnicodeCharacter{210C}{\mathfrak{H}}
\DeclareUnicodeCharacter{2111}{\mathfrak{I}}
\DeclareUnicodeCharacter{211C}{\mathfrak{R}}
\DeclareUnicodeCharacter{2128}{\mathfrak{Z}}
\DeclareUnicodeCharacter{212D}{\mathfrak{C}}
\DeclareUnicodeCharacter{2102}{\mathbb{C}}
\DeclareUnicodeCharacter{210D}{\mathbb{H}}
\DeclareUnicodeCharacter{2115}{\mathbb{N}}
\DeclareUnicodeCharacter{2119}{\mathbb{P}}
\DeclareUnicodeCharacter{211A}{\mathbb{Q}}
\DeclareUnicodeCharacter{211D}{\mathbb{R}}
\DeclareUnicodeCharacter{2124}{\mathbb{Z}}
\DeclareUnicodeCharacter{210B}{\mathcal{H}}
\DeclareUnicodeCharacter{2110}{\mathcal{I}}
\DeclareUnicodeCharacter{2112}{\mathcal{L}}
\DeclareUnicodeCharacter{2113}{\ell}
\DeclareUnicodeCharacter{2118}{\mathcal{P}}
\DeclareUnicodeCharacter{211B}{\mathcal{R}}
\DeclareUnicodeCharacter{212C}{\mathcal{B}}
\DeclareUnicodeCharacter{212F}{\mathcal{e}}
\DeclareUnicodeCharacter{2130}{\mathcal{E}}
\DeclareUnicodeCharacter{2131}{\mathcal{F}}
\DeclareUnicodeCharacter{2133}{\mathcal{M}}
\DeclareUnicodeCharacter{2134}{\mathcal{o}}
\DeclareUnicodeCharacter{2160}{\Roman{1}}
\DeclareUnicodeCharacter{2161}{\Roman{2}}
\DeclareUnicodeCharacter{2162}{\Roman{3}}
\DeclareUnicodeCharacter{2163}{\Roman{4}}
\DeclareUnicodeCharacter{2164}{\Roman{5}}
\DeclareUnicodeCharacter{2165}{\Roman{6}}
\DeclareUnicodeCharacter{2166}{\Roman{7}}
\DeclareUnicodeCharacter{2167}{\Roman{8}}
\DeclareUnicodeCharacter{2168}{\Roman{9}}
\DeclareUnicodeCharacter{2169}{\Roman{10}}
\DeclareUnicodeCharacter{216A}{\Roman{11}}
\DeclareUnicodeCharacter{216B}{\Roman{12}}
\DeclareUnicodeCharacter{216C}{\Roman{50}}
\DeclareUnicodeCharacter{216D}{\Roman{100}}
\DeclareUnicodeCharacter{216E}{\Roman{500}}
\DeclareUnicodeCharacter{216F}{\Roman{1000}}
\DeclareUnicodeCharacter{2170}{\roman{1}}
\DeclareUnicodeCharacter{2171}{\roman{2}}
\DeclareUnicodeCharacter{2172}{\roman{3}}
\DeclareUnicodeCharacter{2173}{\roman{4}}
\DeclareUnicodeCharacter{2174}{\roman{5}}
\DeclareUnicodeCharacter{2175}{\roman{6}}
\DeclareUnicodeCharacter{2176}{\roman{7}}
\DeclareUnicodeCharacter{2177}{\roman{8}}
\DeclareUnicodeCharacter{2178}{\roman{9}}
\DeclareUnicodeCharacter{2179}{\roman{10}}
\DeclareUnicodeCharacter{217A}{\roman{11}}
\DeclareUnicodeCharacter{217B}{\roman{12}}
\DeclareUnicodeCharacter{217C}{\roman{50}}
\DeclareUnicodeCharacter{217D}{\roman{100}}
\DeclareUnicodeCharacter{217E}{\roman{500}}
\DeclareUnicodeCharacter{217F}{\roman{1000}}
\DeclareUnicodeCharacter{2190}{\leftarrow}
\DeclareUnicodeCharacter{2191}{\uparrow}
\DeclareUnicodeCharacter{2192}{\rightarrow}
\DeclareUnicodeCharacter{2193}{\downarrow}
\DeclareUnicodeCharacter{2194}{\leftrightarrow}
\DeclareUnicodeCharacter{2195}{\updownarrow}
\DeclareUnicodeCharacter{2196}{\nwarrow}
\DeclareUnicodeCharacter{2197}{\nearrow}
\DeclareUnicodeCharacter{2198}{\searrow}
\DeclareUnicodeCharacter{2199}{\swarrow}
\DeclareUnicodeCharacter{21A6}{\mapsto}
\DeclareUnicodeCharacter{21D0}{\Leftarrow}
\DeclareUnicodeCharacter{21D2}{\Rightarrow}
\DeclareUnicodeCharacter{21D4}{\iff}
\DeclareUnicodeCharacter{21DD}{\rightsquigarrow}
\DeclareUnicodeCharacter{2200}{\forall}
\DeclareUnicodeCharacter{2202}{\partial}
\DeclareUnicodeCharacter{2203}{\exists}
\DeclareUnicodeCharacter{2204}{\nexists}
\DeclareUnicodeCharacter{2205}{\varnothing}
\DeclareUnicodeCharacter{2208}{\in}
\DeclareUnicodeCharacter{2209}{\notin}
\DeclareUnicodeCharacter{220B}{\ni}
\DeclareUnicodeCharacter{220C}{\not\ni}
\DeclareUnicodeCharacter{220F}{\prod}
\DeclareUnicodeCharacter{2210}{\coprod}
\DeclareUnicodeCharacter{2211}{\sum}
\DeclareUnicodeCharacter{2213}{\mp}
\DeclareUnicodeCharacter{2215}{/}
\DeclareUnicodeCharacter{2216}{\setminus}
\DeclareUnicodeCharacter{2217}{\ast}
\DeclareUnicodeCharacter{2218}{\circ}
\DeclareUnicodeCharacter{2219}{\bullet}
\DeclareUnicodeCharacter{221E}{\infty}
\DeclareUnicodeCharacter{2223}{\mid}
\DeclareUnicodeCharacter{2224}{\nmid}
\DeclareUnicodeCharacter{2225}{\parallel}
\DeclareUnicodeCharacter{2226}{\nparallel}
\DeclareUnicodeCharacter{2227}{\wedge}
\DeclareUnicodeCharacter{2228}{\vee}
\DeclareUnicodeCharacter{2229}{\cap}
\DeclareUnicodeCharacter{222A}{\cup}
\DeclareUnicodeCharacter{2260}{\ne}
\DeclareUnicodeCharacter{2261}{\equiv}
\DeclareUnicodeCharacter{2262}{\not\equiv}
\DeclareUnicodeCharacter{2264}{\leq}
\DeclareUnicodeCharacter{2265}{\geq}
\DeclareUnicodeCharacter{2248}{\approx}
\DeclareUnicodeCharacter{2282}{\subset}
\DeclareUnicodeCharacter{2283}{\supset}
\DeclareUnicodeCharacter{2284}{\not\subset}
\DeclareUnicodeCharacter{2285}{\not\supset}
\DeclareUnicodeCharacter{2286}{\subseteq}
\DeclareUnicodeCharacter{2287}{\supseteq}
\DeclareUnicodeCharacter{2288}{\nsubseteq}
\DeclareUnicodeCharacter{2289}{\nsupseteq}
\DeclareUnicodeCharacter{228A}{\subsetneq}
\DeclareUnicodeCharacter{228B}{\supsetneq}
\DeclareUnicodeCharacter{228F}{\sqsubset}
\DeclareUnicodeCharacter{2290}{\sqsupset}
\DeclareUnicodeCharacter{2291}{\sqsubseteq}
\DeclareUnicodeCharacter{2292}{\sqsupseteq}
\DeclareUnicodeCharacter{2293}{\sqcap}
\DeclareUnicodeCharacter{2294}{\sqcup}
\DeclareUnicodeCharacter{2295}{\oplus}
\DeclareUnicodeCharacter{2296}{\ominus}
\DeclareUnicodeCharacter{2297}{\otimes}
\DeclareUnicodeCharacter{2298}{\oslash}
\DeclareUnicodeCharacter{2299}{\odot}
\DeclareUnicodeCharacter{229A}{\circledcirc}
\DeclareUnicodeCharacter{229B}{\circledast}
\DeclareUnicodeCharacter{229C}{\circledequal}
\DeclareUnicodeCharacter{229D}{\circleddash}
\DeclareUnicodeCharacter{22A2}{\vdash}
\DeclareUnicodeCharacter{22A3}{\dashv}
\DeclareUnicodeCharacter{22A4}{\top}
\DeclareUnicodeCharacter{22A5}{\bot}
\DeclareUnicodeCharacter{22A7}{\models}
\DeclareUnicodeCharacter{22AC}{\nvdash}
\DeclareUnicodeCharacter{22B2}{\vartriangleleft}
\DeclareUnicodeCharacter{22B3}{\vartriangleright}
\DeclareUnicodeCharacter{22C0}{\bigwedge}
\DeclareUnicodeCharacter{22C1}{\bigvee}
\DeclareUnicodeCharacter{22C2}{\bigcap}
\DeclareUnicodeCharacter{22C3}{\bigcup}
\DeclareUnicodeCharacter{22C5}{\cdot}
\DeclareUnicodeCharacter{22C6}{\star}
\DeclareUnicodeCharacter{22EE}{\vdots}
\DeclareUnicodeCharacter{22EF}{\cdots}
\DeclareUnicodeCharacter{22F0}{\udots}
\DeclareUnicodeCharacter{22F1}{\ddots}
\DeclareUnicodeCharacter{27E6}{\llbracket}
\DeclareUnicodeCharacter{27E7}{\rrbracket}
\DeclareUnicodeCharacter{27E8}{\langle}
\DeclareUnicodeCharacter{27E9}{\rangle}
\DeclareUnicodeCharacter{27F8}{\Longleftarrow}
\DeclareUnicodeCharacter{27F9}{\Longrightarrow}
\DeclareUnicodeCharacter{2983}{\lBrace}
\DeclareUnicodeCharacter{2984}{\rBrace}
\DeclareUnicodeCharacter{1D62}{_i}
\DeclareUnicodeCharacter{1D63}{_r}
\DeclareUnicodeCharacter{1D64}{_u}
\DeclareUnicodeCharacter{1D65}{_v}
\DeclareUnicodeCharacter{1D66}{_\beta}
\DeclareUnicodeCharacter{1D67}{_\gamma}
\DeclareUnicodeCharacter{1D68}{_\rho}
\DeclareUnicodeCharacter{1D69}{_\phi}
\DeclareUnicodeCharacter{1D6A}{_\chi}
\DeclareUnicodeCharacter{2090}{_a}
\DeclareUnicodeCharacter{2091}{_e}
\DeclareUnicodeCharacter{2092}{_o}
\DeclareUnicodeCharacter{2093}{_x}
\DeclareUnicodeCharacter{2095}{_h}
\DeclareUnicodeCharacter{2096}{_k}
\DeclareUnicodeCharacter{2097}{_l}
\DeclareUnicodeCharacter{2098}{_m}
\DeclareUnicodeCharacter{2099}{_n}
\DeclareUnicodeCharacter{209A}{_p}
\DeclareUnicodeCharacter{209B}{_s}
\DeclareUnicodeCharacter{209C}{_t}
\DeclareUnicodeCharacter{2C7C}{_j}
\DeclareUnicodeCharacter{2070}{^0}
\DeclareUnicodeCharacter{00B9}{^1}
\DeclareUnicodeCharacter{00B2}{^2}
\DeclareUnicodeCharacter{00B3}{^3}
\DeclareUnicodeCharacter{2074}{^4}
\DeclareUnicodeCharacter{2075}{^5}
\DeclareUnicodeCharacter{2076}{^6}
\DeclareUnicodeCharacter{2077}{^7}
\DeclareUnicodeCharacter{2078}{^8}
\DeclareUnicodeCharacter{2079}{^9}
\DeclareUnicodeCharacter{207A}{^+}
\DeclareUnicodeCharacter{207B}{^-}
\DeclareUnicodeCharacter{207C}{^=}
\DeclareUnicodeCharacter{207D}{^(}
\DeclareUnicodeCharacter{207E}{^)}
\DeclareUnicodeCharacter{2071}{^i}
\DeclareUnicodeCharacter{207F}{^n}
\DeclareUnicodeCharacter{2080}{_0}
\DeclareUnicodeCharacter{2081}{_1}
\DeclareUnicodeCharacter{2082}{_2}
\DeclareUnicodeCharacter{2083}{_3}
\DeclareUnicodeCharacter{2084}{_4}
\DeclareUnicodeCharacter{2085}{_5}
\DeclareUnicodeCharacter{2086}{_6}
\DeclareUnicodeCharacter{2087}{_7}
\DeclareUnicodeCharacter{2088}{_8}
\DeclareUnicodeCharacter{2089}{_9}
\DeclareUnicodeCharacter{208A}{_+}
\DeclareUnicodeCharacter{208B}{_-}
\DeclareUnicodeCharacter{208C}{_=}
\DeclareUnicodeCharacter{208D}{_(}
\DeclareUnicodeCharacter{208E}{_)}

\makeatletter
\let\c@algorithm\relax
\let\c@figure\relax
\makeatother 
\newaliascnt{figure}{float}
\newaliascnt{algorithm}{float}

\algrenewcommand\algorithmicrequire{\textbf{Input:}}
\algrenewcommand\algorithmicensure{\textbf{Output:}}

\theoremstyle{definition}
\declaretheorem[qed=$\Box$, name=Definition]{definition}
\declaretheorem[qed=$\Box$, name=Example, sibling=definition]{example}
\declaretheorem[qed=$\blacksquare$, name=Observation, sibling=definition]{observation}
\declaretheorem[qed=$\blacksquare$, name=Proof idea, style=remark, numbered=no]{proofidea}
\declaretheorem[name=Corollary, refname={Corollary,Corollaries}, sibling=definition]{corollary}
\declaretheorem[name=Theorem, sibling=definition]{theorem}
\declaretheorem[name=Proposition, sibling=definition]{proposition}
\declaretheorem[name=Lemma, sibling=definition]{lemma}

\DeclareMathSymbol{:}{\mathpunct}{operators}{"3A}

\SetSymbolFont{stmry}{bold}{U}{stmry}{m}{n}

\setlist[enumerate,1]{label=(\roman*)}

\newlist{pcases}{enumerate}{5}
\setlist[pcases,1]{
  align=left,
  labelsep=0pt,
  leftmargin=1em,
  labelwidth=0pt,
  parsep=0pt,
  font=\itshape,
  labelsep=1em,
  label={Case~\arabic*:},
  ref={\arabic*}
}
\setlist[pcases,2]{
  align=left,
  labelsep=0pt,
  leftmargin=1em,
  labelwidth=0pt,
  parsep=0pt,
  font=\itshape,
  labelsep=1em,
  label={Case~\arabic{pcasesi}.\arabic*:},
  ref={\arabic{pcasesi}.\arabic*}
}
\setlist[pcases,3]{
  align=left,
  labelsep=0pt,
  leftmargin=1em,
  labelwidth=0pt,
  parsep=0pt,
  font=\itshape,
  labelsep=1em,
  label={Case~\arabic{pcasesi}.\arabic{pcasesii}.\arabic*:},
  ref={\arabic{pcasesi}.\arabic{pcasesii}.\arabic*}
}
\setlist[pcases,4]{
  align=left,
  labelsep=0pt,
  leftmargin=1em,
  labelwidth=0pt,
  parsep=0pt,
  font=\itshape,
  labelsep=1em,
  label={Case~\arabic{pcasesi}.\arabic{pcasesii}.\arabic{pcasesiii}.\arabic*:},
  ref={\arabic{pcasesi}.\arabic{pcasesii}.\arabic{pcasesiii}.\arabic*}
}
\setlist[pcases,5]{
  align=left,
  labelsep=0pt,
  leftmargin=1em,
  labelwidth=0pt,
  parsep=0pt,
  font=\itshape,
  labelsep=1em,
  label={Case~\arabic{pcasesi}.\arabic{pcasesii}.\arabic{pcasesiii}.\arabic{pcasesiv}.\arabic*:},
  ref={\arabic{pcasesi}.\arabic{pcasesii}.\arabic{pcasesiii}.\arabic{pcasesiv}.\arabic*}
}

\crefname{algorithm}{Alg.}{Algs.}
\Crefname{algorithm}{Algorithm}{Algorithms}
\crefname{conjecture}{Conj.}{Conjs.}
\Crefname{conjecture}{Conjecture}{Conjectures}
\crefname{construction}{Con.}{Cons.}
\Crefname{construction}{Construction}{Constructions}
\crefname{corollary}{Cor.}{Cors.}
\Crefname{corollary}{Corollary}{Corollaries}
\crefname{definition}{Def.}{Defs.}
\Crefname{definition}{Definition}{Definitions}
\crefname{example}{Ex.}{Exs.}
\Crefname{example}{Example}{Examples}
\crefname{figure}{Fig.}{Figs.}
\Crefname{figure}{Figure}{Figures}
\crefname{lemma}{Lem.}{Lems.}
\Crefname{lemma}{Lemma}{Lemmas}
\crefname{theorem}{Thm.}{Thms.}
\Crefname{theorem}{Theorem}{Theorems}
\crefname{observation}{Obs.}{Obs.}
\Crefname{observation}{Observation}{Observations}
\crefname{proposition}{Prop.}{Props.}
\Crefname{proposition}{Proposition}{Propositions}
\crefname{section}{Sec.}{Secs.}
\Crefname{section}{Section}{Sections}

\newcommand\comp{\mathop{;}}
\newcommand\dom{\mathrm{dom}}
\newcommand\img{\mathrm{img}}
\newcommand\Runs{\mathrm{R}}
\newcommand\wt{\mathrm{wt}}
\renewcommand\det{\mathrm{det}}
\newcommand\app[2]{#1\ifthenelse{\isempty{#2}}{}{(#2)}}

\DeclareMathOperator\tsbottom{bottom}
\DeclareMathOperator\down{down}
\DeclareMathOperator\equals{top}

\DeclareMathOperator\pop{pop}
\DeclareMathOperator\push{push}
\DeclareMathOperator\up{up}

\newcommand\parto{\mathrel{\dashrightarrow}}

\title{\texorpdfstring{Approximation of \\ Weighted Automata with Storage}{Approximation of Weighted Automata with Storage}}
\newcommand{\titlerunning}{Approximation of Weighted Automata with Storage}

\author{%
  Tobias Denkinger %
  \institute{%
    Faculty of Computer Science, %
    Technische Universität Dresden, \\ %
    Nöthnitzer Str. 46, 01062 Dresden, Germany \\ %
    \email{tobias.denkinger@tu-dresden.de}%
  }%
}
\newcommand{\authorrunning}{Tobias Denkinger}

\begin{document}
\maketitle

\begin{abstract}
  We use a non-deterministic variant of storage types to develop a framework for the approximation of automata with storage.
  This framework is used to provide automata-theoretic views on the approximation of multiple context-free languages and on coarse-to-fine parsing.
\end{abstract}

\section{Introduction}

Formal grammars (e.g. context-free grammars) are used to model natural languages.
Language models are often incorporated into systems that have to guarantee a certain response time, e.g. translation systems or speech recognition systems.
The desire for low response times and the high parsing complexity of the used formal grammars are at odds.
Thus, in real-world applications, the language model is often replaced by another language model that is easier to parse but still captures the desired natural language reasonably well.
This new language model is called an \emph{approximation} of the original language model.
Nederhof~\cite{Ned00a} gives an overview for the approximation of context-free grammars.

In order to approximate a context-free grammar it is common (but not exclusive \cite{Ned00,Cha+06}) to first construct an equivalent pushdown automaton and then approximate this automaton \cite{KraTom81,Pul86,LanLan87,BerSch90,PerWri91,Eva97,Joh98}, e.g. by restricting the height of the pushdown.
Automata with storage \cite{Sco67,Gol79,Eng86,Eng14} generalise pushdown automata.
By attaching weights to the transitions of an automaton with storage, we can model, e.g. the \emph{multiplicity} with which a word belongs to a language or the \emph{cost} of recognising a word~\cite{Sch62,Eil74}.
The resulting devices are called \emph{weighted} automata with storage and were studied in recent literature \cite{HerVog15,VogDroHer16}.
Multiple context-free languages (MCFLs) \cite{SekMatFujKas91,VijWeiJos87} are currently studied as language models because they can express the non-projective constituents and discontinuous dependencies that occur in natural languages \cite{Mai10,KuhSat09}.
Their approximation was recently investigated from a grammar-centric viewpoint \cite{BurLju05,Cra12}.
MCFLs can be captured by automata with specific storage \cite{Vil02,Den16}, which allows an automata-theoretic view on their approximation.

We develop a framework to study the approximation of weighted automata with arbitrary storage.
To deal with non-determinism that arises due to approximation, we use automata with \emph{data storage} \cite{Gol79} which allow instructions to be non-deterministic;\footnote{We add predicates to Goldstine's original definition of data storage.  This does not increase their expressiveness (\cref{lem:predicate-free-normal-form}).} and we investigate their relation to automata with storage (\cref{sec:nd-storage}).
Weighted automata with data storage differ from Engelfriet's automata with storage~\cite{Eng86,Eng14} in two aspects: As instructions we allow binary relations instead of partial functions and each transition is associated with a weight from a semiring.
Using a powerset construction, we show that (weighted) automata with data storage have the same expressive power as (weighted) automata with storage (\cref{lem:implementing-of-nd-storage,lem:implementing-of-nd-storage-weighted}).
Our formalisation of strategies for approximating data storage (called \emph{approximation strategies}) is inspired by the storage simulation of Hoare~\cite{Hoa72,EngVog86}.
We use partial functions as approximation strategies (\cref{sec:approximation}).
Properties of the approximation strategy imply properties of the while approximation process:
If an approximation strategy is a total function, then we have a superset approximation (\cref{thm:superset-approximation,thm:weighted-approx-types}\ref{item:wt-approx-types:over}).
If an approximation strategy is injective, then we have a subset approximation (\cref{thm:subset-approximation,thm:weighted-approx-types}\ref{item:wt-approx-types:under}).
In contrast to Engelfriet and Vogler~\cite{EngVog86}, we do not utilise flowcharts in our constructions.
We demonstrate the benefit of our framework by providing an automata-based view on the approximation of MCFLs
(\cref{sec:approximation-mcfl}) and by describing an algorithm for coarse-to-fine parsing of weighted automata with data storage (\cref{sec:coarse-to-fine-parsing}).

\section{Preliminaries}

The set $\{0, 1, 2, … \}$ of \emph{natural numbers} is denoted by $ℕ$, $ℕ ∖ \{0\}$ is denoted by $ℕ_+$, and $\{1, …, k\}$ is denoted by $[k]$ for every $k ∈ ℕ$ (note that \([0] = ∅\)).
Let $A$ be a set.
The \emph{power set of $A$} is denoted by $\mathcal{P}(A)$.

Let $A$, $B$, and $C$ be sets and let $r ⊆ A × B$ and $s ⊆ B × C$ be binary relations.
We denote $\{(b, a) ∈ B × A ∣ (a, b) ∈ r \}$ by $r^{-1}$, $\{ b ∈ B ∣ (a, b) ∈ r \}$ by $r(a)$ for every $a ∈ A$, and $⋃_{a ∈ A'} r(a)$ by $r(A')$ for every $A' ⊆ A$.
The \emph{sequential composition of $r$ and $s$} is the binary relation
\(
  r \comp s = \{ (a, c) ∈ A × C ∣ ∃b ∈ B: ((a, b) ∈ r) ∧ ((b, c) ∈ s) \}\text{.}
\)
We call $r$ an \emph{endorelation (on $A$)} if $A = B$.
%
A \emph{semiring} is an algebraic structure $(K, {+}, {⋅}, 0, 1)$ where $(K, {+}, 0)$ is a commutative monoid, $(K, {⋅}, 1)$ is a monoid, 0 is absorptive with respect to ${⋅}$, and ${⋅}$ distributes over ${+}$.
We say that $K$ is \emph{complete} if it has a sum operation $∑_I: K^I → K$ that extends ${+}$ for each countable set~$I$ \cite[Sec.~2]{DroKui09}.
Let~$≤$ be a partial order on~$K$.
We say that \emph{$K$ is positively~${≤}$-ordered} if
  $+$ preserves $≤$ (i.e. for each $a, b, c ∈ K$ with $a ≤ b$ holds $a + c ≤ b + c$),
  $⋅$ preserves $≤$ (i.e. for each $a, b, c ∈ K$ with $a ≤ b$ holds $a ⋅ c ≤ b ⋅ c$ and $c ⋅ a ≤ c ⋅ b$), and
  $0 ≤ a$ for each $a ∈ K$ (cf. Droste and Kuich~\cite[Sec.~2]{DroKui09}).

The \emph{set of partial functions from $A$ to $B$} is denoted by $A \parto B$.
The \emph{set of (total) functions from $A$ to $B$} is denoted by $A → B$.
Let $f: A \parto B$ be a partial function.
The \emph{domain of $f$} and the \emph{image of $f$} are defined by
$\dom(f) = \{ a ∈ A ∣ ∃b ∈ B: f(a) = b \}$ and
$\img(f) = \{ b ∈ B ∣ ∃a ∈ A: f(a) = b \}$, respectively.
Abusing the notation, we may sometimes write $f(a) = \text{undefined}$ to denote that $a ∉ \dom(f)$.
Note that every total function is a partial function and that each partial function is a binary relation.

\section{Automata with data storage}
\label{sec:nd-storage}

In addition to the finite state control, automata with storage are allowed to check and manipulate a storage configuration that comes from a possibly infinite set.
We propose a syntactic extension of automata with storage where the set of unary functions (the \emph{instructions}) is replaced by a set of binary relations on the storage configurations.

\subsection{Data storage}

\begin{definition}
  A \emph{data storage} is a tuple \(S = (C, P, R, c_{\text{i}})\) where
    $C$ is a set (of \emph{storage configurations}),
    \(P ⊆ \mathcal{P}(C)\) (\emph{predicates}),
    \(R ⊆ \mathcal{P}(C × C)\) (\emph{instructions}),
    \(c_{\text{i}} ∈ C\) (\emph{initial storage configuration}), and
    the set $r(c)$ is finite for every $r ∈ R$ and $c ∈ C$.
\end{definition}
Our definition of data storage differs from the original definition \cite[Def.~3.1]{Gol79} in that we have predicates.
The “data storage types” introduced by Herrmann and Vogler \cite[Sec.~3]{HerVog16} are similar to our data storages.
For instructions they use partial functions that may depend on the input of the automaton in addition to the current storage configuration instead of binary relations on storage configurations.

Consider a data storage \(S = (C, P, R, c_{\text{i}})\).
If every element of $R$ is a partial function, we call $S$ \emph{deterministic}.
The definition of “deterministic data storage” in this paper coincides with the definition of “storage type” in previous literature \cite{HerVog15,VogDroHer16}.

\begin{example}
  The deterministic data storage $\textrm{Count}$ models simple counting (Engelfriet~\cite[Def.~3.4]{Eng86,Eng14}):
  \(\mathrm{Count} = (ℕ, \{ℕ, ℕ₊, \{0\}\}, \{\mathrm{inc}, \mathrm{dec}\}, 0)\) where
  \(\mathrm{inc} = \{(n, n+1) ∣ n ∈ ℕ\}\) and
  \(\mathrm{dec} = \mathrm{inc}^{-1}\).
\end{example}
  
\begin{example}\label{ex:pd}
  The following deterministic data storage models pushdown storage:%
  \footnote{We allows (in comparison to Engelfriet~\cite[Def.~3.2]{Eng86,Eng14}) the execution of (some) instructions on the empty pushdown.}
  \(\mathrm{PD}_Γ = (Γ^*, P_{\text{pd}}, R_{\text{pd}}, ε)\) where
  \(Γ\) is a nonempty finite set (\emph{pushdown symbols});
  \(P_{\text{pd}} = \{Γ^*, \mathrm{bottom}\} ∪ \{\mathrm{top}_γ ∣ γ ∈ Γ\}\) with
  \( \mathrm{bottom} = \{ ε \} \) and \( \mathrm{top}_γ = \{γw ∣ w ∈ Γ^*\} \) for every $γ ∈ Γ$; and
  $R_{\text{pd}} = \{\mathrm{stay}, \mathrm{pop}\} ∪ \{\mathrm{push}_γ ∣ γ ∈ Γ\} ∪ \{ \mathrm{stay}_γ ∣ γ ∈ Γ\}$ with
  \( \mathrm{stay} = \{(w, w) ∣ w ∈ Γ^*\} \),
  \( \mathrm{pop} = \{(γw, w) ∣ w ∈ Γ^*, γ ∈ Γ\} \),
  \( \mathrm{push}_γ = \{(w, γw) ∣ w ∈ Γ^*\} \), and
  \( \mathrm{stay}_γ = \{(γ'w, γw) ∣ w ∈ Γ^*, γ' ∈ Γ\} \) for every $γ ∈ Γ$.
\end{example}

We call a data storage $S = (C, P, R, c_{\text{i}})$ \emph{boundedly non-deterministic (short: boundedly nd)} if there is a natural number $k$ such that $\lvert r(c) \rvert ≤ k$ holds for every $r ∈ R$ and $c ∈ C$.
  The following two examples illustrate that each deterministic data storage is also boundedly nd, but not vice versa.

\begin{example}\label{ex:pd-with-pop-star}
  $\mathrm{PD}_Γ'$ extends $\mathrm{PD}_Γ$ (cf. \cref{ex:pd}) by adding an instruction $\mathrm{pop}^*$ that allows us to remove arbitrarily many symbols from the top of the pushdown:
  \(\mathrm{PD}_Γ' = (Γ^*, P_{\mathrm{pd}}, R_{\mathrm{pd}} ∪ \{\mathrm{pop}^*\}, ε)\) where
  \(\mathrm{pop}^* = \{(uw, w) ∣ u, w ∈ Γ^*\}\).
  
  The tuple $\mathrm{PD}_Γ'$ is a data storage because
  $\lvert \mathrm{stay}(w) \rvert = 1$,
  $\lvert \mathrm{pop}(w) \rvert ≤ 1$,
  $\lvert \mathrm{push}_γ(w) \rvert = 1$,
  $\lvert \mathrm{stay}_γ(w) \rvert ≤ 1$, and
  $\lvert \mathrm{pop}^*(w) \rvert = \lvert w \rvert + 1$
  for each $w ∈ Γ^*$ and $γ ∈ Γ$ are all finite.
  But $\mathrm{PD}_Γ'$ is not boundedly nd.
  Assume that it were.
  Then there would be a number $k ∈ ℕ$ such that $\lvert r(w) \rvert ≤ k$ for every $r ∈ R_{\text{pd}}$ and $w ∈ Γ^*$.
  But if we take some $w' ∈ Γ^*$ of length $k$, then $\lvert \mathrm{pop}^*(w') \rvert = k + 1 > k$ which contradicts our assumption.
\end{example}

\begin{example}\label{ex:pd-with-push-prime}
   The data storage $\mathrm{PD}_Γ''$ extends $\mathrm{PD}_Γ$ (cf. \cref{ex:pd}) by adding an instruction $\mathrm{push}_Γ$ that allows us to add an  arbitrary symbol from $Γ$ the top of the pushdown:
  \(\mathrm{PD}_Γ'' = (Γ^*, P_{\mathrm{pd}}, R_{\mathrm{pd}} ∪ \{\mathrm{push}_Γ\}, ε)\) where
  \(\mathrm{push}_Γ = \{(w, wγ) ∣ w ∈ Γ^*, γ ∈ Γ\}\).
  
    The data storage $\mathrm{PD}_Γ''$ is boundedly nd because if we take bound $k = \lvert Γ \rvert$, then
    $\lvert \mathrm{stay}(w) \rvert = 1 ≤ k$,
    $\lvert \mathrm{pop}(w) \rvert ≤ 1 ≤ k$,
    $\lvert \mathrm{push}_γ(w) \rvert = 1 ≤ k$,
    $\lvert \mathrm{stay}_γ(w) \rvert ≤ 1 ≤ k$, and
    $\lvert \mathrm{push}_Γ(w) \rvert = \lvert Γ \rvert ≤ k$.
    In particular, if $\lvert Γ \rvert > 1$, then $\mathrm{PD}_Γ''$ is not deterministic because $\lvert \mathrm{push}_Γ(w) \rvert = \lvert Γ \rvert > 1$.
\end{example}

\subsection{Automata with data storage}

\begin{quote}
\emph{For the rest of this paper let $Σ$ be an arbitrary non-empty finite set.}
\end{quote}

\begin{definition}
  Let \(S = (C, P, R, c_{\text{i}})\) be a data storage.
  An \emph{$(S, Σ)$-automaton} is a tuple \(ℳ = (Q, T, Q_{\text{i}}, Q_{\text{f}})\) where
    $Q$ is a finite set (of \emph{states}),
    $T$ is a finite subset of \(Q × (Σ ∪ \{ε\}) × P × R × Q\) (\emph{transitions}),
    \(Q_{\text{i}} ⊆ Q\) (\emph{initial states}), and
    \(Q_{\text{f}} ⊆ Q\) (\emph{final states}).
\end{definition}
  Let \(ℳ = (Q, T, Q_{\text{i}}, Q_{\text{f}})\) be an $(S, Σ)$-automaton and \(S = (C, P, R, c_{\text{i}})\).
  An \emph{$ℳ$-configuration} is an element of \(Q × C × Σ^*\).
  For every \(τ = (q, v, p, r, q') ∈ T\), the \emph{transition relation of $τ$} is the endorelation $⊢_τ$ on the set of $ℳ$-configurations that contains \((q, c, vw) ⊢_τ (q', c', w)\) for every $w ∈ Σ^*$ and $(c, c') ∈ r$ with $c ∈ p$.
  The \emph{run relation of $ℳ$} is $⊢_ℳ = ⋃_{τ ∈ T} {⊢_τ}$.
  The transition relations are extended to sequences of transitions by setting \({⊢_{τ_1⋯τ_k}} = {⊢_{τ_1}} \comp … \comp {⊢_{τ_k}}\) for every \(k ∈ ℕ\) and \(τ_1, …, τ_k ∈ T\).
  In particular, for the case $k =0$ we use the identity on $Q × C × Σ^*$: ${⊢_ε} = \{ (d, d) ∣ d ∈ Q × C × Σ^* \}$.
  The \emph{set of runs of $ℳ$} is the set
  \begin{equation}
    \Runs_ℳ = \big\{ θ ∈ T^* ∣ ∃q, q' ∈ Q, c, c' ∈ C, w, w' ∈ Σ^*: (q, c, w) ⊢_θ (q', c', w') \big\}\text{.}\label{eq:runs}
  \end{equation}
  Let $w ∈ Σ^*$.
  The \emph{set of runs of $ℳ$ on $w$} is
  \(
    \Runs_ℳ(w) = \big\{ θ ∈ T^* ∣ ∃q ∈ Q_{\text{i}}, q' ∈ Q_{\text{f}}, c' ∈ C: (q, c_{\text{i}}, w) ⊢_θ (q', c', ε)\big\}\text{.}\label{eq:runs-word}
  \)
  The \emph{language accepted by $ℳ$} is the set \( L(ℳ) = \{w ∈ Σ^* ∣ \Runs_ℳ(w) ≠ ∅ \}\).
  Let $S$ be a data storage and $L ⊆ Σ^*$.
  We call $L$ \emph{$(S, Σ)$-recognisable} if there is an $(S, Σ)$-automaton $ℳ$ with $L = L(ℳ)$.

\begin{figure}[t]
  \centering
  \begin{tikzpicture}[->, >=stealth', auto, node distance=9em]
    \node[initial  , state] (1)             {1};
    \node[state]            (2) [right of=1] {2};
    \node[accepting, state] (3) [right of=2] {3};
    
    \path (1) edge [loop above] node {$\text{a}, Γ^*                    , \mathrm{push}_Γ$} (1)
              edge [loop below] node {$\text{b}, Γ^*                    , \mathrm{push}_Γ$} (1)
              edge              node {$\#      , Γ^*                    , \mathrm{stay}$}   (2)
          (2) edge [loop above] node {$\text{a}', \mathrm{top}_{\text{a}}, \mathrm{pop}$}    (2)
              edge [loop below] node {$\text{b}', \mathrm{top}_{\text{b}}, \mathrm{pop}$}    (2)
              edge              node {$ε       , \mathrm{bottom}        , \mathrm{stay}$}   (3);
  \end{tikzpicture}  
  \caption{Graph of the $(\mathrm{PD}_Γ'', Σ)$-automaton $ℳ$ from \cref{ex:automaton-with-storage}}
  \label{fig:automaton-with-storage}
\end{figure}

\begin{example}\label{ex:automaton-with-storage}
  Recall the data storage $\mathrm{PD}_Γ''$ from \cref{ex:pd-with-push-prime}.
  Let $Σ = \{\text{a}, \text{b}, \#, \text{a}', \text{b}'\}$ and $Γ = \{\text{a}, \text{b}\}$, and consider the $(\mathrm{PD}_Γ'', Σ)$-automaton $ℳ = ([3], T, \{1\}, \{3\})$ where
  \begin{align*}
    T:\quad
      &\begin{array}[t]{@{(}l@{,\,}l@{,\,}l@{,\,}l@{,\,}l@{)}}
      1 & \text{a}  & Γ^*                     & \mathrm{push}_Γ & 1 \\
      2 & \text{a}' & \mathrm{top}_{\text{a}}   & \mathrm{pop}    & 2
    \end{array}
      &&\begin{array}[t]{@{(}l@{,\,}l@{,\,}l@{,\,}l@{,\,}l@{)}}
      1 & \text{b}  & Γ^*                     & \mathrm{push}_Γ & 1 \\
      2 & \text{b}' & \mathrm{top}_{\text{b}}   & \mathrm{pop}    & 2
    \end{array}
      &&\begin{array}[t]{@{(}l@{,\,}l@{,\,}l@{,\,}l@{,\,}l@{)}l}
      1 & \#       & Γ^*                      & \mathrm{stay}   & 2 \\
      2 & ε        & \mathrm{bottom}          & \mathrm{stay}   & 3 & \text{.} 
    \end{array}
  \end{align*}
  The graph of $ℳ$ is shown in \cref{fig:automaton-with-storage}.
  The label of each edge in the graph contains the input that is read by the corresponding transition, the predicate that is checked, and the instruction that is executed.
  The language recognised by $ℳ$ is
  \( L(ℳ) = \{ u \# v ∣ u ∈ \{\text{a}, \text{b}\}^*, v ∈ \{\text{a}', \text{b}'\}^*, \lvert u \rvert = \lvert v \rvert\}\).
  The automaton $ℳ$ recognises a given word $u \# v$ (with $u ∈ \{\text{a}, \text{b}\}^*$ and $v ∈ \{\text{a}', \text{b}'\}^*$) as follows:
    In state 1, it reads the prefix $u$ and constructs any element of $Γ^*$ of length $\lvert u \rvert$ on the pushdown non-deterministically.
    It then reads $\#$ and goes to state 2.
    In state 2, it reads $\text{a}'$ for each $\text{a}$ on the pushdown and it reads $\text{b}'$ for each $\text{b}$ on the pushdown until the pushdown is empty.
    Since the pushdown can contain any sequence over $\{\text{a}, \text{b}\}$ of length $\lvert u \lvert$, $ℳ$ can read any sequence of $\{\text{a}', \text{b}'\}$ of length $\lvert u \rvert$, ensuring that $\lvert u \rvert = \lvert v \rvert$.
\end{example}

We call a data storage $S = (C, P, R, c_{\text{i}})$ \emph{predicate-free} if $P = \{ C \}$.\footnote{Even though $S$ has a predicate $C$, we still call it predicate-free since $C$ is trivial, i.e. $C$ accepts any storage configuration.}
The following lemma shows that predicate-free-ness is a normal form among data storages.

\begin{lemma}\label{lem:predicate-free-normal-form}
  For every data storage $S$ there is a predicate-free data storage $S'$ such that the classes of $(S, Σ)$-recognisable languages and the class of $(S', Σ)$-recognisable languages are the same.
\end{lemma}
\begin{proof}[Proof idea]
  Encode the predicates of $S$ in the instructions of $S'$.
\end{proof}

\begin{proposition}\label{lem:implementing-of-nd-storage}
  For every data storage $S$ there is a deterministic data storage $\det(S)$ such that the class of $(S, Σ)$-recognisable languages is equal to the class of $(\det(S), Σ)$-recognisable languages.
\end{proposition}
\begin{proof}
  Due to \cref{lem:predicate-free-normal-form} we can assume that $S$ is predicate-free.
  Thus, let \(S = (C, \{ C \}, R, c_{\text{i}})\).
  Using a power set construction, we obtain the deterministic data storage \(\det(S) = (\mathcal{P}(C), \{ \mathcal{P}(C) \}, \det(R), \{c_{\text{i}}\})\) where
    \(\det(R) = \{ \det(r) ∣ r ∈ R \}\) with \(\det(r) = \{ (d, r(d)) ∣ d ⊆ C, r(d) ≠ ∅\}\) for every $r ∈ R$.

  Let \(ℳ = (Q, T, Q_{\text{i}}, Q_{\text{f}})\) be an $(S, Σ)$-automaton and  \(ℳ' = (Q, T', Q_{\text{i}}, Q_{\text{f}})\) be a $(\det(S), Σ)$-automaton.
  We say that $ℳ$ and $ℳ'$ are \emph{related} if $T' = \det(T) = \{\det(τ) ∣ τ ∈ T\}$ with $\det(τ) = (q, v, \mathcal{P}(C), \det(r), q')$ for each $τ = (q, v, C, r, q') ∈ T$.
  Clearly, for every $(S, Σ)$-automaton there is an $(\det(S), Σ)$-automaton such that both are related, and vice versa.

  Now let \(ℳ = (Q, T, Q_{\text{i}}, Q_{\text{f}})\) be an $(S, Σ)$-automaton and \(ℳ' = (Q, \det(T), Q_{\text{i}}, Q_{\text{f}})\) be a $(\det(S), Σ)$-automaton.
  Note that $ℳ$ and $ℳ'$ are related.
  We extend $\det: T → \det(T)$ to a function $\det: T^* → (\det(T))^*$ by point-wise application.
  We can show for every $θ ∈ T^*$ by induction on the length of $θ$ that
  \begin{equation}
      ∀ q, q' ∈ Q, c, c' ∈ C, w, w' ∈ Σ^*: \quad
      (q, c, w) ⊢_θ (q', c', w') \iff
      ∀d ∋ c: ∃d' ∋ c': (q, d, w) ⊢_{\det(θ)} (q', d', w')
    \label{eq:implementing-of-nd-storage:IH}
  \end{equation}
  holds.
 We obtain $L(ℳ) = L(ℳ')$ from \eqref{eq:implementing-of-nd-storage:IH} and since \(\{c_{\text{i}}\}\) is the initial storage configuration of $ℳ'$.
\end{proof}

For practical reasons it might be preferable to avoid the construction of power sets.
The proof of the following \nameCref{prop:implementing-nd-storage-no-powerset} shows a construction for boundedly nd data storages.

\begin{proposition}\label{prop:implementing-nd-storage-no-powerset}
  Let $S = (C, P, R, c_{\text{i}})$ be a boundedly nd data storage.
  There is a deterministic data storage $S'$ with the same set of storage configurations such that the class of $(S, Σ)$-recognisable languages is contained in the class of $(S', Σ)$-recognisable languages.
\end{proposition}
\begin{proof}
  We construct the deterministic data storage \(S' = (C, P, R', c_{\text{i}})\) where $R'$ is constructed as follows:
  Let $r ∈ R$ and \(r(c)_1, …, r(c)_{m_{r, c}}\) be a fixed enumeration of the elements of $r(c)$ for every $c ∈ C$.
  Furthermore, let \(k = \max\{\lvert r(c) \rvert ∣ r ∈ R, c ∈ C\}\).
  Since $S$ is boundedly nd, the number $k$ is well defined.
  We define for each $i ∈ [k]$ an instruction \(r_i'\) by \(r_i'(c) = r(c)_i\) if $i ≤ m_{r, c}$ and $r_i'(c) = \text{undefined}$ otherwise.
  Let~$R'$ contain the instruction $r_i'$ for every $r ∈ R$ and $i ∈ [k]$.
  Now let \(ℳ = (Q, T, Q_{\text{i}}, Q_{\text{f}})\) be an $(S, Σ)$-automaton.
  We construct the $(S', Σ)$-automaton \(ℳ' = (Q, T', Q_{\text{i}}, Q_{\text{f}})\) where $T'$ contains for every transition $t = (q, v, p, r, q') ∈ T$ and $i ∈ [k]$ the transition \(t_i' = (q, v, p, r_i', q')\).
  Then
  \(
    {⊢_ℳ}
    = ⋃\nolimits_{t ∈ T} {⊢_t}
    = ⋃\nolimits_{t = (q, v, p, r, q') ∈ T} ⋃\nolimits_{i ∈ [k]} {⊢_{t_i'}}
    = ⋃\nolimits_{t' ∈ T'} {⊢_{t'}}
    = {⊢_{ℳ'}}
  \)
  and thus \(L(ℳ) = L(ℳ')\).
\end{proof}

The above construction fails for data storages that are not boundedly nd.
Consider the data storage \(\mathrm{PD}_Γ'\) from \cref{ex:pd-with-pop-star}.
Then there exists no bound $k_{\mathrm{pop}^*} ∈ ℕ$ as would be required in the proof.

The containment shown in \cref{prop:implementing-nd-storage-no-powerset} is strict as the following example reveals.

  

\begin{example}[due to Nederhof~\cite{Ned17pc}]\label{ex:implementing-nd-storage-no-powerset-proper-containment}
  Recall the data storage \(\mathrm{PD}_Γ''\) from \cref{ex:pd-with-push-prime}.
  Consider the similar data storage $\mathrm{PD}_Γ^† = (Γ^*, \{ Γ^*, \mathrm{bottom} \}, \{ \mathrm{stay}, \mathrm{push}_Γ \} ∪ \{\mathrm{pop}_γ ∣ γ ∈ Γ\}, ε)$ where
  \(\mathrm{pop}_γ = \{(γw, w) ∣ γ ∈ Γ, w ∈ Γ^*\}\) for each $γ ∈ Γ$.
  We can again think of $Γ^*$ as a pushdown.
  Now, starting from $\mathrm{PD}_Γ^†$, we construct the deterministic data storage $(\mathrm{PD}^†_Γ)'$ by the construction given in \cref{prop:implementing-nd-storage-no-powerset}.
  We thereby obtain $(\mathrm{PD}^†_Γ)' = (Γ^*, \{Γ^*, \mathrm{bottom}\}, \{ \mathrm{stay}\} ∪ \{ \mathrm{push}_γ ∣ γ ∈ Γ \} ∪ \{ \mathrm{pop}_γ ∣ γ ∈ Γ\}, ε)$.
  The only difference between $\mathrm{PD}_Γ^†$ and $(\mathrm{PD}^†_Γ)'$ is that the instruction \( \mathrm{push}_Γ \) is replaced by the $\lvert Γ \rvert$ instructions in the set $\{ \mathrm{push}_γ ∣ γ ∈ Γ \}$.

  Now consider the sets $Σ = \{ \text{a}, \text{b} \}$ and $Γ = Σ$, and the language $L = \{ww^{\text{R}} ∣ w ∈ Σ^*\} ⊆ Σ^*$ where $w^{\text{R}}$ denotes the reverse of $w$ for each $w ∈ Σ^*$.
  The following $((\mathrm{PD}^†_Γ)', Σ)$-automaton $ℳ'$ recognises $L$ and thus demonstrates that $L$ is $((\mathrm{PD}^†_Γ)', Σ)$-recognisable:
  $ℳ' = ([3], T', \{1\}, \{3\})$ with
  \begin{align*}
    T':\quad
    &\begin{array}[t]{@{(}l@{,\,}l@{,\,}l@{,\,}l@{,\,}l@{)}}
      1 & \text{a} & Γ^*             & \mathrm{push}_{\text{a}} & 1 \\
      2 & \text{a} & Γ^*             & \mathrm{pop}_{\text{a}}  & 2
    \end{array}
    &&\begin{array}[t]{@{(}l@{,\,}l@{,\,}l@{,\,}l@{,\,}l@{)}}
      1 & \text{b} & Γ^*             & \mathrm{push}_{\text{b}} & 1 \\
      2 & \text{b} & Γ^*             & \mathrm{pop}_{\text{b}}  & 2
    \end{array}
    &&\begin{array}[t]{@{(}l@{,\,}l@{,\,}l@{,\,}l@{,\,}l@{)}l}
      1 & ε        & Γ^*             & \mathrm{stay}     & 2 \\
      2 & ε        & \mathrm{bottom} & \mathrm{stay}     & 3 & \text{.} 
    \end{array}
  \end{align*}
  In state 1, $ℳ'$ stores the input in reverse on the pushdown until it decides non-deterministically go to state 2.
  In state 2, $ℳ$ accepts the sequence of symbols that is stored on the pushdown.
  We can only enter the final state 3 if the pushdown is empty, thus $ℳ'$ recognises $L$.

  On the other hand, there is no $(\mathrm{PD}_Γ^†, Σ)$-automaton $ℳ$ that recognises $L$.
  Assume that some $(\mathrm{PD}_Γ^†, Σ)$-automaton $ℳ$ recognises $L$.
  Then $ℳ$ would have to encode the first half of the input in the pushdown since this unbounded information can not be stored in the states.
  The only instruction that adds information to the pushdown is $\mathrm{push}_Γ$.
  Thus, in the first half of the input, whenever we read the symbol a, we have to execute $\mathrm{push}_Γ$; and whenever we read the symbol b, we also have to execute $\mathrm{push}_Γ$.
  This offers no means of distinguishing the two situations (reading symbol a and reading symbol b) and hence no means of encoding the first half of the input in the pushdown.
\end{example}

\begin{proposition}\label{obs:equivalent-fsa}
  Let $S = (C, P, R, c_{\text{i}})$ be a data storage and $L$ be an $(S, Σ)$-recognisable language.
  If $C$ is finite, then $L$ is recognisable (by a finite state automaton). 
\end{proposition}
\begin{proof}
  We will use a product construction.
  In particular, the states of the constructed finite state automaton are elements of $Q × C$.
  For this we employ non-deterministic finite-state automata with extended transition function (short: fsa) from Hopcroft and Ullman~\cite[Sec.~2.3]{HopUll79} in a notation similar to that of automata with storage. (We simply leave out the storage-related parts of the transitions.)

  Let \(ℳ = (Q, T, Q_{\text{i}}, Q_{\text{f}})\).
  We construct the fsa \(ℳ' = (Q × C, Σ, T', Q_{\text{i}} × \{c_{\text{i}}\}, Q_{\text{f}} × C)\) where
  \(T' = \{ ((q, c), v, (q', c')) ∣ (q, v, p, r, q') ∈ T, (c, c') ∈ r, c ∈ p \}\).
  We can show
  \begin{equation}
    ∀q, q' ∈ Q, c, c' ∈ C, w, w' ∈ Σ^*: \quad
    (q, c, w) ⊢_ℳ^* (q', c', w') \iff ((q, c), w) ⊢_{ℳ'}^* ((q', c'), w')\text{.}
    \label{eq:equivalent-fsa:IH}
  \end{equation}
  by straight-forward induction on the length of runs.
  Using \eqref{eq:runs} and \eqref{eq:equivalent-fsa:IH}, we then derive $L(ℳ) = L(ℳ')$.
\end{proof}

\section{Approximation of automata with data storage}
\label{sec:approximation}

An approximation strategy maps a data storage to another data storage.
It is specified in terms of storage configurations and naturally extended to predicates and instructions.

\begin{definition}
  Let $S = (C, P, R, c_{\text{i}})$ be a data storage.
  An \emph{approximation strategy} is a partial function \(A: C \parto C'\) for some set $C'$.
  We call $A$ \emph{$S$-proper} if $(A^{-1} \comp r \comp A)(c')$ is finite for every $r ∈ R$ and $c' ∈ C'$.
\end{definition}

\begin{definition}\label{def:storage-approximation}
  Let \(S = (C, P, R, c_{\text{i}})\) be a data storage and \(A: C \parto C'\) be an $S$-proper approximation strategy.
  The \emph{approximation of $S$ with respect to $A$} is the data storage $\app{A}{S} = (C', \app{A}{P}, \app{A}{R}, A(c_{\text{i}}))$ where
    \(\app{A}{P} = \{\app{A}{p} ∣ p ∈ P\}\) with \(\app{A}{p} = \{ A(c) ∣ c ∈ p \}\) for every $p ∈ P$, and
    \(\app{A}{R} = \{\app{A}{r} ∣ r ∈ R\}\) with \(\app{A}{r} = A^{-1} \comp r \comp A\) for every $r ∈ R$.
\end{definition}

\begin{example}
  Consider the approximation strategy \(A_{\mathrm{o}}: ℕ → \{\mathrm{odd}\} ∪ \{2n ∣ n ∈ ℕ\}\) that assigns to every odd number the value $\mathrm{odd}$ and to every even number the number itself.
  Then $A_{\mathrm{o}}$ \emph{is not} $\mathrm{Count}$-proper since \((A_{\mathrm{o}}^{-1} \comp \mathrm{inc} \comp A_{\mathrm{o}})(\mathrm{odd}) = (A_{\mathrm{o}}^{-1} \comp \mathrm{dec} \comp A_{\mathrm{o}})(\mathrm{odd}) = \{2n ∣ n ∈ ℕ\}\) is not finite.

  On the other hand, consider the approximation strategy \(A_{\mathrm{eo}}: ℕ → \{\mathrm{even}, \mathrm{odd}\}\) that returns $\mathrm{odd}$ for every odd number and $\mathrm{even}$ otherwise.
  Then $A_{\mathrm{eo}}$ \emph{is} $\mathrm{Count}$-proper since
  \((A_{\mathrm{eo}}^{-1} \comp \mathrm{inc} \comp A_{\mathrm{eo}})(\mathrm{even}) = \{\mathrm{odd}\} = (A_{\mathrm{eo}}^{-1} \comp \mathrm{dec} \comp A_{\mathrm{eo}})(\mathrm{even})\) and
  \((A_{\mathrm{eo}}^{-1} \comp \mathrm{inc} \comp A_{\mathrm{eo}})(\mathrm{odd}) = \{\mathrm{even}\} = (A_{\mathrm{eo}}^{-1} \comp \mathrm{dec} \comp A_{\mathrm{eo}})(\mathrm{odd})\) are finite.
\end{example}

\begin{definition}\label{def:automaton-approximation}
  Let \(ℳ = (Q, T, Q_{\text{i}}, Q_{\text{f}})\) be an $(S, Σ)$-automaton and $A$ an $S$-proper approximation strategy.
  The \emph{approximation of $ℳ$ with respect to $A$} is the $(\app{A}{S}, Σ)$-automaton \(\app{A}{ℳ} = (Q, \app{A}{T}, Q_{\text{i}}, Q_{\text{f}})\) where \(\app{A}{T} = \{\app{A}{τ} ∣ τ ∈ T\}\) and $\app{A}{τ} = (q, v, \app{A}{p}, \app{A}{r}, q')$ for each \(τ = (q, v, p, r, q') ∈ T\).
\end{definition}


\begin{example}\label{ex:Count-approximation-Aeo}
  Let $Σ = \{\text{a}, \text{b}\}$.
  Consider the ($\mathrm{Count}, Σ)$-automaton $ℳ = ([3], T, \{1\}, \{3\})$ and its approximation $\app{A_{\text{eo}}}{ℳ} = ([3], \app{A_{\text{eo}}}{T}, \{1\}, \{3\})$ with
  \begin{align*}
    T:
    \begin{array}[t]{r@{{}=(}l@{,\,}l@{,\,}l@{,\,}l@{,\,}l@{)}}
      τ_1 & 1 & \text{a} & ℕ   & \mathrm{inc} & 1 \\
      τ_2 & 1 & \text{b} & ℕ   & \mathrm{dec} & 2 \\
      τ_3 & 2 & \text{b} & ℕ   & \mathrm{dec} & 2 \\
      τ_4 & 2 & ε        & \{0\} & \mathrm{inc} & 3
    \end{array}
    &&\app{A_{\text{eo}}}{T}:
    \begin{array}[t]{r@{{}=(}l@{,\,}l@{,\,}l@{,\,}l@{,\,}l@{)}}
      τ_1' & 1 & \text{a} & \app{A_{\textrm{eo}}}{ℕ}   & \app{A_{\textrm{eo}}}{\mathrm{inc}} & 1 \\
      τ_2' & 1 & \text{b} & \app{A_{\textrm{eo}}}{ℕ}   & \app{A_{\textrm{eo}}}{\mathrm{dec}} & 2 \\
      τ_3' & 2 & \text{b} & \app{A_{\textrm{eo}}}{ℕ}   & \app{A_{\textrm{eo}}}{\mathrm{dec}} & 2 \\
      τ_4' & 2 & ε        & \app{A_{\textrm{eo}}}{\{0\}} & \app{A_{\textrm{eo}}}{\mathrm{inc}} & 3
    \end{array}
  \end{align*}
  where $\app{A_{\textrm{eo}}}{ℕ} = \app{A_{\textrm{eo}}}{ℕ₊} = \{\mathrm{even}, \mathrm{odd}\}$ and $\app{A_{\textrm{eo}}}{\{0\}} = \{\mathrm{even}\}$ are the predicates of $\app{A_{\textrm{eo}}}{\mathrm{Count}}$, and $\app{A_{\textrm{eo}}}{\mathrm{inc}} = \app{A_{\textrm{eo}}}{\mathrm{dec}} = \{(\mathrm{even}, \mathrm{odd}), (\mathrm{odd}, \mathrm{even})\}$ is the instruction of $\app{A_{\textrm{eo}}}{\mathrm{Count}}$.
  The word $\text{aabb} ∈ \{\text{a}, \text{b}\}^*$ is recognised by both automata:
  \[
    \begin{array}{rlllll}
      (1, 0, \text{aabb})
      &⊢_{τ_1} (1, 1, \text{abb})
      &⊢_{τ_1} (1, 2, \text{bb})
      &⊢_{τ_2} (2, 1, \text{b})
      &⊢_{τ_3} (2, 0, ε)
      &⊢_{τ_4} (3, 1, ε) \\[.2em]
      (1, \mathrm{even}, \text{aabb})
      &⊢_{τ_1'} (1, \mathrm{odd}, \text{abb})
      &⊢_{τ_1'} (1, \mathrm{even}, \text{bb})
      &⊢_{τ_2'} (2, \mathrm{odd}, \text{b})
      &⊢_{τ_3'} (2, \mathrm{even}, ε)
      &⊢_{τ_4'} (3, \mathrm{odd}, ε)\text{.}
    \end{array}
  \]
  On the other hand, the word $\text{bb}$ can be recognised by $\app{A_{\textrm{eo}}}{ℳ}$ but not by $ℳ$:
  \begin{align*}
    (1, \mathrm{even}, \text{bb})
    ⊢_{τ_2'} (2, \mathrm{odd}, \text{b})
    ⊢_{τ_3'} (2, \mathrm{even}, ε)
    ⊢_{τ_4'} (3, \mathrm{odd}, ε)\text{.}
    \tag*\qedhere
  \end{align*}
\end{example}

\begin{observation}\label{obs:composition-of-approximations}
  Let $S = (C, P, R, c_{\text{i}})$, $ℳ$ be an $(S, Σ)$-automaton, and $A_1: C \parto \bar{C}$ and $A_2: \bar{C} \parto C'$ be approximation strategies.
  If $A_1$ is $S$-proper and $A_2$ is $\app{A_1}{S}$-proper, then $\app{A_2}{\app{A_1}{ℳ}} = \app{(A_1 \comp A_2)}{ℳ}$.
\end{observation}

We call an approximation strategy \emph{total} if it is a total function and we call it \emph{injective} if it is an injective partial function.
The distinction between \emph{total} and \emph{injective} approximation strategies allows us to define two preorders on approximation strategies (\cref{def:finer-approximation}) and provides us with simple criteria to ensure that an approximation strategy leads to a superset (\cref{thm:superset-approximation}) or a subset approximation (\cref{thm:subset-approximation}).

\begin{definition}\label{def:finer-approximation}
  Let \(A_1: C \parto C₁\) and \(A_2: C \parto C₂\) be approximation strategies.
  We call $A_1$ \emph{finer than} $A_2$, denoted by $A_1 \preceq A_2$, if there is a total approximation strategy $A: C₁ → C₂$ with $A_1 \comp A = A_2$.
  We call $A_1$ \emph{less partial than} $A_2$, denoted by $A_1 ⊑ A_2$, if there is an injective approximation strategy $A: C₁ \parto C₂$ with $A_1 \comp A = A_2$.
\end{definition}

\subsection{Superset approximations}

In this section we will show that total approximation strategies (i.e. total functions) lead to superset approximations.

\begin{lemma}\label{lem:superset-approximation}
  Let $ℳ = (Q, T, Q_{\text{i}}, Q_{\text{f}})$ be an $(S, Σ)$-automaton, $S = (C, P, R, c_{\text{i}})$, and $A$ be an $S$-proper total approximation strategy.
  We extend $\app{A}{}: T → \app{A}{T}$ to sequences of transitions by point-wise application.
  Then for each $θ ∈ T^*$, $q, q' ∈ Q$, $c, c' ∈ C$, $w, w' ∈ Σ^*$:
  \(
    (q, c, w) ⊢_θ (q', c', w') ⟹ (q, A(c), w) ⊢_{\app{A}{θ}} (q', A(c'), w')\text{.}
  \)
\end{lemma}
\begin{proofidea}
  The claim can be shown by straightforward induction on the length of $θ$.
\end{proofidea}

\begin{theorem}\label{thm:superset-approximation}
  Let \(ℳ\) be an $(S, Σ)$-automaton and $A$ be an $S$-proper total approximation strategy.
  Then \(L(\app{A}{ℳ}) ⊇ L(ℳ)\).
\end{theorem}
\begin{proof}
  The claim follows immediately from \cref{lem:superset-approximation} and the definition of $\app{A}{ℳ}$.
\end{proof}

\begin{example}
  Recall $ℳ$ and $\app{A_{\text{eo}}}{ℳ}$ from \cref{ex:Count-approximation-Aeo}.
  Their recognised languages are $L(ℳ) = \{\text{a}^n\text{b}^n ∣ n ∈ ℕ_+\}$ and $L(\app{A_{\text{eo}}}{ℳ}) = \{ \text{a}^m \text{b}^n ∣ m ∈ ℕ,n ∈ ℕ_+, m ≡ n \mod 2\}$.
  Thus $L(\app{A_{\text{eo}}}{ℳ})$ is a superset of $L(ℳ)$.
\end{example}

\begin{corollary}\label{prop:finer-approximation}
  Let $ℳ$ be an $(S, Σ)$-automaton, and $A_1$ and $A_2$ be $S$-proper approximation strategies.
  If $A_1$ is finer than $A_2$, then $L(\app{A_1}{ℳ}) ⊆ L(\app{A_2}{ℳ})$.
\end{corollary}
\begin{proof}
  Since $A_1$ is finer than $A_2$, there is a total approximation strategy $A$ such that $A_1 \comp A = A_2$.
  It follows from the fact that $A_2$ is $S$-proper and from $A_1 \comp A = A_2$ that $A$ must be $\app{A₁}{S}$-proper.
  Hence we obtain
  \(
    L(\app{A_1}{ℳ})
    \stackrel{\text{\cref{thm:superset-approximation}}}{⊆}
    L\big(\app{A}{\app{A_1}{ℳ}}\big)
    \stackrel{\text{\cref{obs:composition-of-approximations}}}{=}
    L(\app{(A_1 \comp A)}{ℳ})
    =
    L(\app{A_2}{ℳ})\text{.}
  \)
\end{proof}

The following example shows four approximation strategies that occur in the literature.
The first three approximation strategies approximate a context-free language by a recognisable language (taken from Nederhof~\cite[Sec.~7]{Ned00}).
The fourth approximation strategy approximates a context-free language by another context-free language.
It is easy to see that the shown approximation strategies are total and thus lead to superset approximations.

\begin{example}\label{ex:superset-approximation}
  Let $Γ$ be a finite set and $k ∈ ℕ_+$.
  \begin{enumerate}
  \item
    Evans~\cite{Eva97} proposed to map each pushdown to its top-most element.
    The same result is achieved by dropping condition~7 and~8 from Baker~\cite{Bak81}.
    This idea is expressed by the total approximation strategy \( A_{\text{top}}: Γ^* → Γ ∪ \{@\} \) with
    \( A_{\text{top}}(ε) = @ \) and \( A_{\text{top}}(γw) = γ \) for every $w ∈ Γ^*$ and $γ ∈ Γ$, where $@$ is a new symbol that is not in $Γ$.
  \item
    Bermudez and Schimpf~\cite{BerSch90} proposed to map each pushdown to its top-most $k$ elements.
    The total approximation strategy \(A_{\text{top}, k}: Γ^* → \{w ∈ Γ^* ∣ \lvert w \vert ≤ k\} \) implements this idea where
    \( A_{\text{top},k}(w) = w \) if $\lvert w \rvert ≤ k$ and
    \( A_{\text{top},k}(w) = u \) if $w$ is of the form $uv$ for some $u ∈ Γ^k$ and $v ∈ Γ^+$.
  \item
    Pereira and Wright~\cite{PerWri91} proposed to map each pushdown to one where no pushdown symbol occurs more than once.
    To achieve this, they replace each substrings of the form $γw'γ$ (for some $γ ∈ Γ$ and $w' ∈ Γ^*$) in the given pushdown by $γ$:
    Consider
    \( A_{\text{uniq}}: Γ^* → \mathrm{Seq}_{\text{nr}}(Γ) \) with
    \( A_{\text{uniq}}(w) = A_{\text{uniq}}(uγv) \) if $w$ is of form $uγw'γv$ for some $γ ∈ Γ$ and
    \( A_{\text{uniq}}(w) = w \) otherwise,
    where $\mathrm{Seq}_{\text{nr}}(Γ)$ denotes the set of all sequences over $Γ$ without repetition.
  \item\label{ex:superset-approximation:equivalent-nts}
    In their coarse-to-fine parsing approach for context-free grammars (short: CFG), Charniak et~al.~\cite{Cha+06} propose, given an equivalence relation~$≡$ on the set of non-terminals $N$ of some CFG $G$, to construct a new CFG $G'$ whose non-terminals are the equivalence classes of $≡$.\footnote{Charniak et~al.~\cite{Cha+06} actually considered probabilistic CFGs, but for the sake of simplicity we leave out the probabilities here.}
    Let $Σ$ be the terminal alphabet of $G$.
    Say that $g: N → N/{≡}$ is the function that assigns for a nonterminal of $G$ its corresponding equivalence class; and let $g': (N ∪ Σ)^* → ((N/{≡}) ∪ Σ)^*$ be an extension of $g ∪ \{(σ, σ) ∣ σ ∈ Σ\}$.
    Then $g'$ is $\mathrm{PD}_{N ∪ Σ}$-proper and \( L(\app{g'}{ℳ}) = L(G') \) where $ℳ$ is the $(\mathrm{PD}_{N ∪ Σ}, Σ)$-automaton obtained from $G$ by the usual construction \cite[Thm.~5.3]{HopUll79}.\qedhere
  \end{enumerate}
\end{example}

\subsection{Subset approximations}

In this section we will show that injective approximation strategies lead to a subset approximation, this is proved by a variation of the proof of \cref{thm:superset-approximation}.

\begin{lemma}\label{lem:subset-approximation}
  Let \(ℳ = (Q, T, Q_{\text{i}}, Q_{\text{f}})\) be an $(S, Σ)$-automaton, \(S = (C, P, R, c_{\text{i}})\), and $A$ be an $S$-proper injective approximation strategy.
  Then for each $θ ∈ T^*$, $q, q' ∈ Q$, $c, c' ∈ \img(A)$, $w, w' ∈ Σ^*$: 
  \(
      (q, c, w) ⊢_{\app{A}{θ}} (q', c', w') ⟹ (q, A^{-1}(c), w) ⊢_θ (q', A^{-1}(c'), w')\text{.}
  \)
\end{lemma}
\begin{proofidea}
  The claim can be shown by straightforward induction on the length of $θ$.
\end{proofidea}

\begin{theorem}\label{thm:subset-approximation}
  Let $ℳ$ be an $(S, Σ)$-automaton and $A$ be an $S$-proper injective approximation strategy.
  Then $L(\app{A}{ℳ}) ⊆ L(ℳ)$.
\end{theorem}
\begin{proof}
  Then the claim follows immediately from \cref{lem:subset-approximation} and the definition of $\app{A}{ℳ}$.
\end{proof}

\begin{corollary}\label{prop:less-partial-approximation}
  Let $ℳ$ be an $(S, Σ)$-automaton and $A_1$ and $A_2$ be $S$-proper approximation strategies.
  If $A_1$ is less partial than $A_2$, then $L(\app{A_1}{ℳ}) ⊇ L(\app{A_2}{ℳ})$.
\end{corollary}
\begin{proof}
  Since $A_1$ is less partial than $A_2$, we know that there is an injective approximation strategy $A$ such that $A_1 \comp A = A_2$.
  As in the proof of \cref{prop:finer-approximation} we know that $A$ is $\app{A_1}{S}$-proper.
  Hence we obtain
  \begin{align*}
    L(\app{A_1}{ℳ})
    \stackrel{\text{\cref{thm:subset-approximation}}}{⊇}
      L\big(\app{A}{\app{A_1}{ℳ}}\big)
    \stackrel{\text{\cref{obs:composition-of-approximations}}}{=}
      L(\app{(A_1 \comp A)}{ℳ})
    =
      L(\app{A_2}{ℳ})\text{.}
    \tag*\qedhere
  \end{align*}
\end{proof}

The following example approximates a context-free language with a recognisable language (taken from Nederhof~\cite[Sec.~7]{Ned00}).
It is easy to see that the shown approximation strategy is injective and thus leads to subset approximations.

\begin{example}\label{ex:subset-approximation}
  Let $Γ$ be a finite set and $k ∈ ℕ_+$.
  Krauwer and des Tombe~\cite{KraTom81}, Pulman~\cite{Pul86}, and Langendoen and Langsam~\cite{LanLan87} proposed to disallow pushdowns of height greater than $k$.
  This can be achieved by the partial identity
  \( A_{\text{bd},k}: Γ^+ \parto \{w ∈ Γ ∣ \lvert w \rvert ≤ k\} \) where
  \( A_{\text{bd},k}(w) = w \) if $\lvert w \rvert ≤ k$ and
  \( A_{\text{bd},k}(w) = \text{undefined} \) if \(\lvert w \rvert > k\).
\end{example}

\subsection{Potentially incomparable approximations}

The following example shows that our framework is also capable of expressing approximation strategies that lead neither to superset nor to subset approximations.

\begin{example}
  Let $Γ$ be a (not necessarily finite) set, $Δ$ be a finite set, $k ∈ ℕ_+$, and $g: Γ → Δ$ be a total function.
  For pushdown automata with an infinite pushdown alphabet, Johnson~\cite[end of Section~1.4]{Joh98} proposed to first approximate the infinite pushdown alphabet with a finite set and then restrict the pushdown height to $k$.
  This can be easily expressed as the composition of two approximations:
  \begin{align*}
    A_{\text{incomp},k}: Γ^+ &\parto \{w ∣ w ∈ Δ, \lvert w \rvert ≤ k\}
    &A_{\text{incomp},k} &= \hat{g} \comp A_{\text{bound},k}
  \end{align*}
  where $\hat{g}: Γ^+ → Δ^+$ is the point-wise application of $g$.
  Let $\lvert Δ \rvert < \lvert Γ \rvert$.
  Then $\hat{g}$ is total but not injective, $A_{\text{bound},k}$ is injective but not total, and $A_{\text{incomp},k}$ is neither total nor injective.
  Hence \cref{thm:superset-approximation,thm:subset-approximation} provide no further insights about the approximation strategy $A_{\text{incomp},k}$.
  This concurs with the observation of Johnson \cite[end of Section~1.4]{Joh98} that $A_{\text{incomp},k}$ is not guaranteed to induce either subset or superset approximations.
\end{example}

\subsection{Approximation of weighted automata with storage}

\begin{definition}\label{def:weighted-automaton:syntax}
  Let $S$ be a data storage and $K$ be a complete semiring.
  An \emph{$(S, Σ, K)$-automaton} is a tuple \(ℳ = (Q, T, Q_{\text{i}}, Q_{\text{f}}, δ)\) where
    $(Q, T, Q_{\text{i}}, Q_{\text{f}})$ is an $(S, Σ)$-automaton and
    \(δ: T → K\) (\emph{transition weights}).
    We sometimes denote $(Q, T, Q_{\text{i}}, Q_{\text{f}})$ by $ℳ_{\text{uw}}$ (“$\text{uw}$” stands for unweighted).
\end{definition}
Consider the $(S, Σ, K)$-automaton \(ℳ = (Q, T, Q_{\text{i}}, Q_{\text{f}}, δ)\).
The \emph{$ℳ$-configurations},
the \emph{run relation of $ℳ$}, and
the \emph{set of runs of $ℳ$ on $w$} for every $w ∈ Σ^*$ are the same as for $ℳ_{\text{uw}}$.
The \emph{weight of $θ$ in $ℳ$} is the value \(\wt_ℳ(θ) = δ(τ_1) ⋅ … ⋅ δ(τ_k)\) for every $θ = τ_1 ⋯ τ_k$ with $τ_1, …, τ_k ∈ T$.  In particular, we let $\wt_ℳ(ε) = 1$.
The \emph{weighted language induced by $ℳ$} is the function $⟦ℳ⟧: Σ^* → K$ where
\begin{equation}
  ⟦ℳ⟧(w) = ∑\nolimits_{θ ∈ \Runs_ℳ(w)} \wt_ℳ(θ) \label{eq:weighted-automaton:semantics}
\end{equation}
For every $w ∈ Σ^*$.
Let $S$ be a data storage, $K$ be a complete semiring, and $r: Σ^* → K$.
We call $r$ \emph{$(S, Σ, K)$-recognisable} if there is an $(S, Σ, K)$-automaton $ℳ$ with $r = ⟦ℳ⟧$.

We extend \cref{lem:implementing-of-nd-storage} to the weighted case, using the functions $\det$ as defined in \cref{lem:implementing-of-nd-storage}.

\begin{proposition}\label{lem:implementing-of-nd-storage-weighted}
  The classes of $(S, Σ, K)$-recognisable and of $(\det(S), Σ, K)$-recog\-nis\-able languages are the same for every data storage $S$ and semiring $K$.
\end{proposition}
\begin{proof}
  Let $ℳ = (Q, T, Q_{\text{i}}, Q_{\text{f}}, δ)$ be an $(S, Σ, K)$-automaton and $ℳ' = (Q', T', Q_{\text{i}}', Q_{\text{f}}', δ')$ a $(\det(S), Σ, K)$-automaton.
  We call $ℳ$ and $ℳ'$ \emph{related} if $ℳ_{\text{uw}}$ and $ℳ_{\text{uw}}'$ are related, and $δ'(\det(τ)) = δ(τ)$ for every $τ ∈ T$.
  Note that $\det\colon T → \det(T)$ is a bijection.
  Clearly, for every $(S, Σ, K)$-automaton $ℳ$ there is an $(\det(S), Σ, K)$-automaton $ℳ'$ such that $ℳ$ and $ℳ'$ are related and vice versa.
  It remains to be shown that $⟦ℳ⟧ = ⟦ℳ'⟧$.
  For every $w ∈ Σ^*$, we derive
  \begin{align*}
    ⟦ℳ⟧(w)
    \stackrel{\text{\eqref{eq:weighted-automaton:semantics}}}{=}
      ∑\nolimits_{θ ∈ R_ℳ} \wt_ℳ(θ)
    = ∑\nolimits_{θ ∈ R_ℳ} \wt_{ℳ'}(\det(θ))
    \stackrel{\text{\eqref{eq:implementing-of-nd-storage:IH}}}{=}
      ∑\nolimits_{θ' ∈ R_{ℳ'}} \wt_{ℳ'}(θ')
    \stackrel{\text{\eqref{eq:weighted-automaton:semantics}}}{=}
      ⟦ℳ'⟧(w)\text{.}
      \tag*{\qedhere}
  \end{align*}
\end{proof}

\begin{definition}\label{def:weighted-approx}
  Let \(ℳ = (Q, T, Q_{\text{i}}, Q_{\text{f}}, δ)\) be an $(S, Σ, K)$-automaton and $A$ be an $S$-proper approximation strategy.
  The \emph{approximation of $ℳ$ with respect to $A$} is the $(\app{A}{S}, Σ, K)$-automaton
  \(
    \app{A}{ℳ} = (Q, \app{A}{T}, Q_{\text{i}}, Q_{\text{f}}, \app{A}{δ})
  \)
  where $\app{A}{S}$ and $\app{A}{T}$ are defined as in \cref{def:automaton-approximation}, and 
  \( \app{A}{δ}(τ') = ∑_{τ ∈ T: \app{A}{τ} = τ'} δ(τ) \)
  for every $τ' ∈ \app{A}{T}$.
\end{definition}

\begin{lemma}\label{lem:wt-approx}
  Let \(ℳ\) be an $(S, Σ, K)$-automaton, $A$ be an $S$-proper approximation strategy, ${≤}$ be a partial order on $K$, and $K$ be positively ${≤}$-ordered.
  \begin{enumerate}
  \item\label{item:wt-approx-types:over}
    \(\wt_{\app{A}{ℳ}}(θ') ≥ ∑_{θ ∈ \Runs_{ℳ} : \app{A}{θ} = θ'} \wt_ℳ(θ)\) for every \(θ' ∈ \Runs_{\app{A}{ℳ}}\).
  \item\label{item:wt-approx-types:under}
    If $A$ is injective, then \(\wt_{\app{A}{ℳ}}(θ') = ∑_{θ ∈ \Runs_{ℳ} : \app{A}{θ} = θ'} \wt_ℳ(θ)\) for every \(θ' ∈ \Runs_{\app{A}{ℳ}}\).
  \end{enumerate}
\end{lemma}
\begin{proof}
  \textbf{ad~\ref{item:wt-approx-types:over}:} We proof the claim by induction on the length of $θ'$.
    For $θ' = ε$, we derive
    \begin{align*}
      \wt_{\app{A}{ℳ}}(ε) = 1 ≥ 1 = \wt_{ℳ}(ε) = ∑\nolimits_{θ ∈ \Runs_ℳ: \app{A}{θ} = ε} \wt_ℳ(θ)\text{.}
    \end{align*}
    For $θ'τ' ∈ \Runs_{\app{A}{ℳ}}$ with $τ' ∈ \app{A}{T}$, we derive
    \begin{align*}
      wt_{\app{A}{ℳ}}(θ'τ')
      &= \wt_{\app{A}{ℳ}}(θ') ⋅ \app{A}{δ}(τ') \\*
      &≥ \big( ∑\nolimits_{θ ∈ \Runs_ℳ, \app{A}{θ} = θ'} \wt_ℳ(θ) \big) ⋅ \app{A}{δ}(τ')
        \tag{by IH and since $⋅$ preserves $≤$} \\
      &= \big( ∑\nolimits_{θ ∈ \Runs_ℳ, \app{A}{θ} = θ'} \wt_ℳ(θ) \big) ⋅ \big( ∑\nolimits_{τ ∈ T: \app{A}{τ} = τ'} δ(τ) \big)
        \tag{by \cref{def:weighted-approx}} \\
      &= ∑\nolimits_{θ ∈ \Runs_ℳ, τ ∈ T: (\app{A}{θ} = θ') ∧ (\app{A}{τ} = τ')} \wt_ℳ(θ) ⋅ δ(τ)
        \tag{by distributivity of $K$} \\
      &≥ ∑\nolimits_{θ ∈ \Runs_ℳ, τ ∈ T: θτ ∈ \Runs_ℳ ∧ (\app{A}{θτ} = θ'τ')} \wt_ℳ(θ) ⋅ δ(τ)
        \tag{by $(*)$ and since $+$ preserves $≤$} \\*
      &= ∑\nolimits_{\bar{θ} ∈ \Runs_ℳ: (\app{A}{\bar{θ}} = θ'τ')} \wt_ℳ(\bar{θ})
        \tag{by \cref{def:weighted-approx}}
    \end{align*}
    For $(*)$, we note that the index set of the left sum subsumes that of the right sum and hence $≥$ is justified.

    \textbf{ad~\ref{item:wt-approx-types:under}:} The proof follows the same structure as the proof of \ref{item:wt-approx-types:over}.
    But we make the following modifications:
    In the induction base, we can write “$=$” instead of “$≥$” since $1 = 1$.
    For the induction step, we assume that \ref{item:wt-approx-types:under} holds for every $θ'$ of length $n$.
    Then the “$≥$” in the second line of the induction step can be replaced by “$=$”.
    In order to turn the “$≥$” in the fifth line of the induction step into “$=$”, we propose that the index sets of the left and the right sum are the same.
    This holds since $A$ is injective, $θ'τ'$ is in $R_{\app{A}{ℳ}}$, and hence (by \cref{lem:subset-approximation}) each $θτ$ with $\app{A}{θτ} = θ'τ'$ is in $\Runs_ℳ$.
\end{proof}

\begin{theorem}\label{thm:weighted-approx-types}
  Let \(ℳ\) be an $(S, Σ, K)$-automaton, $A$ be an $S$-proper approximation strategy, and ${≤}$ be a partial order on $K$, and $K$ be positively ${≤}$-ordered.
  \begin{enumerate}
  \item\label{item:weighted-approx-types:over}
    If $A$ is total, then \(⟦\app{A}{ℳ}⟧(w) ≥ ⟦ℳ⟧(w)\) for every $w ∈ Σ^*$.
  \item\label{item:weighted-approx-types:under}
    If $A$ is injective, then \(⟦\app{A}{ℳ}⟧(w) ≤ ⟦ℳ⟧(w)\) for every $w ∈ Σ^*$.
  \end{enumerate}
\end{theorem}
\begin{proof}
  \textbf{ad~\ref{item:weighted-approx-types:over}:} For every $w ∈ Σ^*$, we derive
    \begin{align*}
      &⟦\app{A}{ℳ}⟧(w)
      \stackrel{\text{\eqref{eq:weighted-automaton:semantics}}}{=}
        ∑\nolimits_{θ' ∈ \Runs_{\app{A}{ℳ}}(w)} \wt_{\app{A}{ℳ}}(θ')
      \stackrel{(*)}{≥}
        ∑\nolimits_{θ' ∈ \Runs_{\app{A}{ℳ}}(w)} ∑\nolimits_{θ ∈ \Runs_{ℳ}\colon \app{A}{θ} = θ'} \wt_{ℳ}(θ) \\*
      &\qquad\stackrel{\text{\cref{def:automaton-approximation}}}{=}
        ∑\nolimits_{θ' ∈ \Runs_{\app{A}{ℳ}}(w)} ∑\nolimits_{θ ∈ \Runs_{ℳ}(w)\colon \app{A}{θ} = θ'} \wt_ℳ(θ)
      \stackrel{(\dagger)}{=}
        ∑\nolimits_{θ ∈ \Runs_{ℳ}(w)} \wt_ℳ(θ)
      \stackrel{\text{\eqref{eq:weighted-automaton:semantics}}}{=}
        ⟦ℳ⟧(w)
    \end{align*}
    where $(*)$ follows from \cref{lem:wt-approx}\,\ref{item:wt-approx-types:over} and the fact that $+$ preserves $≤$.
    For $(\dagger)$, we argue that for each $θ ∈ \Runs_ℳ(w)$ there is exactly one $θ' ∈ R_{\app{A}{ℳ}}(w)$ with $\app{A}{θ} = θ'$ since $A$ is total.
    Hence the left side and the right side of the equation have exactly the same addends.
    Then, since $+$ is commutative, the “$=$” is justified.

  \textbf{ad~\ref{item:weighted-approx-types:under}:} For every $w ∈ Σ^*$, we derive
    \begin{align*}
      &⟦\app{A}{ℳ}⟧(w)
      \stackrel{\text{\eqref{eq:weighted-automaton:semantics}}}{=}
        ∑\nolimits_{θ' ∈ \Runs_{\app{A}{ℳ}}(w)} \wt_{\app{A}{ℳ}}(θ')
      \stackrel{\text{\cref{lem:wt-approx}\,\ref{item:wt-approx-types:under}}}{=}
        ∑\nolimits_{θ' ∈ \Runs_{\app{A}{ℳ}}(w)} ∑\nolimits_{θ ∈ \Runs_{ℳ}\colon \app{A}{θ} = θ'} \wt_{ℳ}(θ) \\*
      &\qquad\stackrel{\text{\cref{def:automaton-approximation}}}{=}
        ∑\nolimits_{θ' ∈ \Runs_{\app{A}{ℳ}}(w)} ∑\nolimits_{θ ∈ \Runs_{ℳ}(w)\colon \app{A}{θ} = θ'} \wt_ℳ(θ)
      \stackrel{(\ddagger)}{≤}
        ∑\nolimits_{θ ∈ \Runs_{ℳ}(w)} \wt_ℳ(θ)
      \stackrel{\text{\eqref{eq:weighted-automaton:semantics}}}{=}
        ⟦ℳ⟧(w)\text{.}
    \end{align*}
    For $(\ddagger)$, we argue that for each $θ ∈ \Runs_ℳ(w)$ there is at most one $θ' ∈ R_{\app{A}{ℳ}}(w)$ with $\app{A}{θ} = θ'$ since $A$ is a partial function.
    Hence all the addends on the left side of the inequality also occur on the right side.
    But there may be an addend $\wt_ℳ(θ)$ on the right side which does not occur on the left side because $\app{A}{θ} = \mathrm{undefined}$.
    Since $+$ preserves $≤$, the “$≤$” is justified. \qedhere
\end{proof}

\section{Approximation of multiple context-free languages}
\label{sec:approximation-mcfl}

Due to the equivalence of pushdown automata and context-free grammars \cite[Thms.~5.3 and~5.4]{HopUll79}, the approximation strategies in \cref{ex:superset-approximation,ex:subset-approximation} can be used for the approximation of context-free languages.
The framework presented in this paper together with the automata characterisation of multiple context-free languages \cite[Thm.~18]{Den16} allows an automata-theoretic view on the approximation of multiple context-free languages.
The automata characterisation uses an excursion-restricted form of automata with tree-stack storage \cite{Den16}.
A tree-stack is a tree with a designated position inside of it (the \emph{stack pointer}).
The automaton can read the label under the stack pointer, can determine whether the stack pointer is at the bottom (i.e. the root), and can modify the tree stack by moving the stack pointer or by adding a node.
The excursion-restriction bounds how often the stack pointer may enter a position from its parent node.

\begin{definition}
  Let $Γ$ be a finite set.
  The \emph{tree-stack storage over $Γ$} is the deterministic data storage \( \mathrm{TSS}_Γ = (\mathrm{TS}_Γ, P_{\text{ts}}, R_{\text{ts}}, c_{\text{i}, \text{ts}}) \) where
  \begin{itemize}
  \item \(\mathrm{TS}_Γ\) is the set of tuples $⟨ξ, ρ⟩$ where $ξ: ℕ_+^* \parto Γ ∪ \{@\}$, $\dom(ξ)$ is finite and prefix-closed,\footnote{A set $D ⊆ ℕ_+^*$ is \emph{prefix closed} if for each $w ∈ D$, every prefix of $w$ is also in $D$.} $ρ ∈ \dom(ξ)$, and $ξ(ρ') = @$ iff $ρ' = ε$ (We call $ξ$ the \emph{stack} and $ρ$ the \emph{stack pointer} of $⟨ξ, ρ⟩$.);
  \item \(c_{\text{i}, \text{ts}} = ⟨\{(ε, @)\}, ε⟩\);
  \item \(P_{\text{ts}} = \{\mathrm{TS}_Γ, \tsbottom\} ∪ \{\equals_γ ∣ γ ∈ Γ\}\) with 
    \(\tsbottom = \{⟨ξ, ρ⟩ ∈ \mathrm{TS}_Γ ∣ ρ = ε\}\) and
    \(\equals_γ = \{⟨ξ, ρ⟩ ∈ \mathrm{TS}_Γ ∣ ξ(ρ) = γ\}\) for every $γ ∈ Γ$; and
  \item \(R_{\text{ts}} = \{\down\} ∪ \{\up_n, \push_{n, γ} ∣ n ∈ ℕ, γ ∈ Γ\}\) where for each $n ∈ ℕ_+$ and $γ ∈ Γ$:
    \begin{itemize}
    \item \(\up_n = \{(⟨ξ, ρ⟩, ⟨ξ, ρn⟩) ∣ ⟨ξ, ρ⟩ ∈ \mathrm{TS}_Γ, ρn ∈ \dom(ξ)\}\),
    \item \(\down = ⋃_{n ∈ ℕ₊} \up_n^{-1}\), and
    \item \(\push_{n, γ} = \{(⟨ξ, ρ⟩, ⟨ξ ∪ \{(ρn, γ)\}, ρn⟩) ∣ ⟨ξ, ρ⟩ ∈ \mathrm{TS}_Γ, ρn ∉ \dom(ξ)\}\).\qedhere
    \end{itemize}
  \end{itemize}
\end{definition}

\begin{example}\label{ex:mcfl}
  Consider $Σ = \{\text{a}, \text{b}, \text{c}\}$, $Γ = \{*, \# \}$, the $(\mathrm{TSS}_Γ, Σ)$-automaton $ℳ = ([4], T, \{1\}, \{4\})$, and
  \begin{align*}
    T:
    \begin{array}[t]{r@{{}=(}l@{,\,}l@{,\,}l@{,\,}l@{,\,}l@{)}}
      τ₁ & 1 & \text{a} & \mathrm{TS}_Γ & \push_{1,*}  & 1 \\ 
      τ₂ & 1 & ε        & \mathrm{TS}_Γ & \push_{1,\#} & 2 \\
      τ₃ & 2 & ε        & \equals_\#    & \down       & 2
    \end{array}
    \quad
    \begin{array}[t]{r@{{}=(}l@{,\,}l@{,\,}l@{,\,}l@{,\,}l@{)}}
      τ₄ & 2 & \text{b} & \equals_*    & \down       & 2 \\
      τ₅ & 2 & ε        & \tsbottom     & \up_1      & 3 \\
      τ₆ & 3 & \text{c} & \equals_*    & \up_1       & 3 
    \end{array}
    \quad
    \begin{array}[t]{r@{{}=(}l@{,\,}l@{,\,}l@{,\,}l@{,\,}l@{)}}
      τ₇ & 3 & ε        & \equals_\#   & \down       & 4\text{.}
    \end{array}
  \end{align*}
  The runs of $ℳ$ all have a specific form:
  $ℳ$ executes $τ₁$ arbitrarily often (say $n$ times) until it executes $τ₂$, leading to the storage configuration
  \( ζ = ⟨\{(ε, @), (1, *), …, (1^n, *), (1^{n+1}, \#)\}, 1^{n+1}⟩ \)
  where $1^k$ means that 1 is repeated $k$ times.
  The stack of $ζ$ is a monadic tree where the leave is labelled with $\#$, the root is labelled with $@$, and the remaining $n$ nodes are labelled with $*$.
  The stack pointer of $ζ$ points to the leave.
  From this configuration $ℳ$ executes $τ₃$ once and $τ₄$ $n$ times (i.e. for each $*$ on the stack), moving the stack pointer to the root.
  Then $ℳ$ executes $τ₅$ once and $τ₆$ $n$ times, leading to the final state.
  Hence the language of $ℳ$ is $L(ℳ) = \{ \text{a}^n \text{b}^n \text{c}^n ∣ n ∈ ℕ \}$, which is not context-free.
\end{example}

\begin{example}\label{ex:approximation-mcfl}
  The following two approximation strategies for multiple context-free languages are taken from the literature.
  Let $Γ$ be a finite set.
  \begin{enumerate}
  \item
    Van Cranenburgh~\cite[Sec.~4]{Cra12} observed that the idea of \cref{ex:superset-approximation}\,\ref{ex:superset-approximation:equivalent-nts} also applies to multiple context-free grammars (short: MCFG).
    The idea can be applied to tree-stack automata similarly to the way it was applied to pushdown automata in \cref{ex:superset-approximation}\,\ref{ex:superset-approximation:equivalent-nts}.
    The resulting data storage is still a tree-stack storage.
    This approximation strategy is total and thus leads to a superset approximation.
  \item
    Burden and Ljunglöf~\cite[Sec.~4]{BurLju05} and van Cranenburgh~\cite[Sec.~4]{Cra12} proposed to split each production of a given MCFG into multiple productions, each of fan-out 1.
    Since the resulting grammar is of fan-out 1, it produces a context-free language and can be recognised by a pushdown automaton.
    The corresponding approximation strategy in our framework is
    \( A_{\text{cf}, Γ}: \mathrm{TS}_Γ → Γ^* \) with  \( A_{\text{cf}, Γ}((ξ, n_1⋯n_k)) = ξ(n_1⋯n_k) ⋯ ξ(n_1n_2) ξ(n_1) \)
    for every $(ξ, n_1⋯n_k) ∈ \mathrm{TS}_Γ$ with $n_1, …, n_k ∈ ℕ_+$.
    The resulting data storage is a pushdown storage.
    $A_{\text{cf}, Γ}$ is total and thus leads to a superset approximation.\qedhere
  \end{enumerate}
\end{example}

\begin{figure}
  \centering
\(  \begin{array}{@{(}l@{,\,}l@{,\,}l@{,\,}l@{,\,}l@{)}}
      1 & \text{a} & Γ^*        & \push_*            & 1 \\
      1 & ε        & Γ^*        & \push_\#           & 2 \\
      2 & ε        & \equals_\# & \pop               & 2 \\
      2 & \text{b} & \equals_*  & \pop               & 2 \\
      2 & ε        & \tsbottom  & \push_* ∪ \push_\# & 3 \\
      3 & \text{c} & \equals_*  & \push_* ∪ \push_\# & 3 \\
      3 & ε        & \equals_\# & \pop               & 4
  \end{array}
\) \qquad\qquad                                       
\(  \begin{array}{@{(}l@{,\,}l@{,\,}l@{,\,}l@{,\,}l@{)}}
      1 & \text{a} & Γ_@    & \{(γ,*) ∣ γ ∈ Γ_@\}           & 1 \\
      1 & ε        & Γ_@    & \{(γ,\#) ∣ γ ∈ Γ_@\}          & 2 \\
      2 & ε        & \{\#\} & \{(γ, γ') ∣ γ, γ' ∈ Γ_@\}     & 2 \\
      2 & \text{b} & \{*\}  & \{(γ, γ') ∣ γ, γ' ∈ Γ_@\}     & 2 \\
      2 & ε        & \{@\}  & \{(γ, γ') ∣ γ ∈ Γ_@, γ' ∈ Γ\} & 3 \\
      3 & \text{c} & \{*\}  & \{(γ, γ') ∣ γ ∈ Γ_@, γ' ∈ Γ\} & 3 \\
      3 & ε        & \{\#\} & \{(γ, γ') ∣ γ, γ' ∈ Γ_@\}     & 4                  
    \end{array}
\)
  \caption{Transitions of \(A_{\text{cf},Γ}(ℳ)\) (left) and \((A_{\text{cf},Γ} \comp A_{\text{top}})(ℳ)\) (right)}
  \label{fig:tss-approximation-transitions}
\end{figure}

\begin{example}\label{ex:approximations-of-mcfl}
  Let us consider the $(\mathrm{TSS}_Γ, Σ)$-automaton $ℳ$ from \cref{ex:mcfl}.
  \Cref{fig:tss-approximation-transitions} shows the transitions of the $(\app{A_{\text{cf},Γ}}{\mathrm{TSS}_Γ}, Σ)$-automaton $\app{A_{\text{cf},Γ}}{ℳ}$ (cf. \cref{ex:approximation-mcfl}) and
  the $(\app{(A_{\text{cf},Γ} \comp A_{\text{top}})}{\mathrm{TSS}_Γ}, Σ)$-automaton $\app{(A_{\text{cf},Γ} \comp A_{\text{top}})}{ℳ}$ (cf. also \cref{ex:superset-approximation}).
  The languages recognised by the two automata are 
    \(L(\app{A_{\text{cf},Γ}}{ℳ}) = \{ a^n b^n c^m ∣ n, m ∈ ℕ \}\) and
    \(L(\app{(A_{\text{cf},Γ} \comp A_{\text{top}})}{ℳ}) = \{a^n b^m c^k ∣ n, m, k ∈ ℕ \}\).
  Clearly, \(L(\app{A_{\text{cf},Γ}}{ℳ})\) is a context-free language.
  Since \(\app{(A_{\text{cf}} \comp A_{\text{top}})}{ℳ}\) has finitely many storage configurations, its language is recognisable by a finite state automaton (\cref{obs:equivalent-fsa}).
\end{example}

\section{Coarse-to-fine $n$-best parsing for weighted automata with storage}
\label{sec:coarse-to-fine-parsing}

Parsing is a process that takes a finite representation $ℛ$ of a language \(L(ℛ) ⊆ Σ^*\) and a word $w ∈ Σ^*$, and outputs analyses of $w$ in $ℛ$.
If $ℛ$ is a grammar, then the analyses of $w$ are the \emph{parse trees} in~$ℛ$ for~$w$.
If $ℛ$ is an automaton (with storage), then the analyses of $w$ are the runs of~$ℛ$ on~$w$.
Since this paper is concerned with weighted automata with storage, let $ℛ$ be an $(S, Σ, K)$-automaton.
Also, let $K$ be partially ordered by a relation $≤$.
We will call a run $θ$ „better than“ a run $θ'$ if $\wt_ℛ(θ) ≥ \wt_ℛ(θ')$.
Using $\wt_ℛ$, we can assign weights to the runs of $ℛ$ on $w$ and enumerate those runs in descending order (with respect to $≤$) of their weights.\footnote{The resulting list of runs is not unique since different runs may get the same weight and since we only have a partial order.}
If we output the first $n$ from the descending list of runs, we call the parsing \emph{$n$-best parsing} \cite{HuaChi05}.

\emph{Coarse-to-fine parsing} \cite{Cha+06} employs a simpler (i.e. easier to parse) automaton $ℛ'$ to parse $w$ and uses the runs of $ℛ'$ on $w$ to narrow the search space for the runs of $ℛ$ on $w$.
To ensure that there are runs of $ℛ'$ on $w$ whenever there are runs of $ℛ$ on $w$, we require that $L(ℛ') ⊇ L(ℛ)$.
The automaton $ℛ'$ is obtained by superset approximation.
In particular, we require $ℛ' = \app{A}{ℛ}$ for some total approximation strategy $A$.

\begin{algorithm}
  \caption{Coarse-to-fine $n$-best parsing for weighted automata with storage}\label{alg:coarse-to-fine}
  \begin{algorithmic}[1]
    \Require
    $(S, Σ, K)$-automaton $ℳ$, \enspace
    $S$-proper total approximation strategy $A$, \enspace
    $n ∈ ℕ$, \enspace
    word $w ∈ Σ^*$
    \Ensure
    some set of $n$ greatest (with respect to the image under $\wt_ℳ$ and $≤$) runs of $ℳ$ on $w$
    \vspace{.5\baselineskip}
    \State \(X ← ∅\)\Comment{$X$ is the set of runs of $ℳ$ on $w$ that were already found}
    \State \(Y ← \Runs_{\app{A}{ℳ}}(w)\)\Comment{$Y$ is the set of runs of $\app{A}{ℳ}$ on $w$ that were not yet considered}
    \While{\(\lvert X \rvert < n\) \textbf{or} \(\min_{θ ∈ X} \wt_{ℳ}(θ) < \max_{θ' ∈ Y} \wt_{\app{A}{ℳ}}(θ')\)}
      \State \(θ' ← \text{smallest element of $Y$ with respect to the image under $\wt_{\app{A}{ℳ}}$}\)
      \State \(Y ← Y ∖ \{θ'\}\)
      \For{\textbf{each} \(θ ∈ \app{A}{}^{-1}(θ')\) that is a sequence of transitions in $ℳ$}
        \If{\(θ ∈ \Runs_ℳ\)}
          \(X ← X ∪ \{θ\}\)\Comment{it is sufficient to only check the storage behaviour for $θ$}
        \EndIf
      \EndFor
    \EndWhile
    \State\Return a set of $n$ greatest elements of $X$ with respect to the image under $\wt_ℳ$
  \end{algorithmic}
\end{algorithm}

\Cref{alg:coarse-to-fine} describes coarse-to-fine $n$-best parsing for weighted automata with storage.
The inputs are
  an $(S, Σ, K)$-automaton $ℳ$,
  an $S$-proper approximation strategy $A$ which will be used to construct an approximation of $ℳ$,
  a natural number $n$ which specifies how many runs should be computed, and
  a word $w ∈ Σ^*$ which we want to parse.
The output is a set of $n$-best runs of $ℳ$ on $w$.
The algorithm starts with a set $X$ that is empty (line~1) and a set $Y$ that contains all the runs of $\app{A}{ℳ}$ on $w$ (line~2).
Then, as long as $X$ has less than $n$ elements or an element of $Y$ is greater than the smallest element in $X$ with respect to their weights (line~3),
  we take the greatest element $θ'$ of $Y$ (line~4),
  remove $θ'$ from $Y$ (line~5),
  calculate the corresponding sequences $θ$ of transitions from $ℳ$ (line~6), and
  add $θ$ to $X$ if $θ$ is a run of $ℳ$ (line~7).
  
We can restrict the automaton $\app{A}{ℳ}$ to the input $w$ with the usual product construction.
The set of runs of the resulting product automaton (let us call it $ℳ_{A, w}$) can be mapped onto $\Runs_{\app{A}{ℳ}}(w)$ by some projection $φ$.
Hence $ℳ_{A, w}$ (finitely) represents $\Runs_{\app{A}{ℳ}}(w)$.
The automaton $ℳ_{A, w}$ can be construed as a (not necessarily finite) graph $G_{A, w}$ with the $ℳ_{A, w}$-configurations as nodes.
The edges shall be labelled with the images of the corresponding transitions of $ℳ_{A, w}$ under $φ$.
Then the paths (i.e. sequences of edge labels) in $G_{A, w}$ from the initial $ℳ_{A, w}$-configuration to all the final $ℳ_{A, w}$-configurations are exactly the elements of $\Runs_{\app{A}{ℳ}}(w)$.
Those paths can be enumerated in descending order of their weights using a variant of Dijkstra's algorithm.
This provides us with a method to compute $\max_{θ' ∈ Y} \wt_{\app{A}{ℳ}}(θ')$ on line~3 and $θ'$ on line~4 of \cref{alg:coarse-to-fine}.

\begin{example}
  Let 
    $Γ = \{\text{a}, \text{b}, \text{c}\}$,
    $Σ = Γ ∪ \{\#\}$,
    $K$ be the Viterbi semiring \((ℕ ∪ \{∞\}, {\min}, {+}, ∞, 0)\) with linear order ${≤}$, and
    $A_\#: Γ^* → ℕ, u ↦ \lvert u \rvert$ be a total approximation strategy.
  Note that $\app{A_\#}{\mathrm{PD}_Γ} = \mathrm{Count}$.
  Now consider
    the \((\mathrm{PD}_Γ, Σ, K)\)-automaton \(ℳ = ([3], T, \{1\}, \{3\}, δ)\) and
    the $(\mathrm{Count}, Σ, K)$-automaton \(\app{A_\#}{ℳ} = ([3], T', \{1\}, \{3\}, δ')\) where
    \(T = \{τ_1, …, τ_8\}\) and \(T' = \{τ_1', τ_2', τ_4', τ_5', τ_6', τ_7', τ_8'\}\) with
  \begin{align*}
    &\begin{array}[t]{@{}r@{{=} (}l@{,\,}l@{,\,}l@{,\,}l@{,\,}l@{)}}
      τ_1 & 1 & \text{a} & Γ^*                    & \mathrm{push}_{\text{a}} & 1 \\
      τ_5 & 2 & \text{a} & \mathrm{top}_{\text{a}} & \mathrm{pop}            & 2
    \end{array}
    &&\begin{array}[t]{@{}r@{{=} (}l@{,\,}l@{,\,}l@{,\,}l@{,\,}l@{)}}
      τ_2 & 1 & ε        & Γ^*                    & \mathrm{push}_{\text{b}} & 1 \\
      τ_6 & 2 & \text{b} & \mathrm{top}_{\text{b}} & \mathrm{pop}_{\text{b}}  & 2
    \end{array}
    &&\begin{array}[t]{@{}r@{{=} (}l@{,\,}l@{,\,}l@{,\,}l@{,\,}l@{)}}
     τ_3 & 1 & ε        & Γ^*                    & \mathrm{push}_{\text{c}} & 1 \\
     τ_7 & 2 & \text{c} & \mathrm{top}_{\text{c}} & \mathrm{pop}            & 2
    \end{array}
    &&\begin{array}[t]{@{}r@{{=} (}l@{,\,}l@{,\,}l@{,\,}l@{,\,}l@{)}}
     τ_4 & 1 & \#       & Γ^*             & \mathrm{stay}     & 2 \\
     τ_8 & 2 & ε        & \mathrm{bottom} & \mathrm{stay}     & 3
    \end{array} \\[.3em]
    &\begin{array}[t]{@{}r@{{=}(}l@{,\,}l@{,\,}l@{,\,}l@{,\,}l@{)}}
      τ_1' & 1 & \text{a} & ℕ   & \mathrm{inc} & 1 \\
      τ_5' & 2 & \text{a} & ℕ_+ & \mathrm{dec} & 2
    \end{array}
    &&\begin{array}[t]{@{}r@{{=}(}l@{,\,}l@{,\,}l@{,\,}l@{,\,}l@{)}}
      τ_2' & 1 & ε        & ℕ   & \mathrm{inc} & 1 \\
      τ_6' & 2 & \text{b} & ℕ_+ & \mathrm{dec} & 2
    \end{array}
    &&\begin{array}[t]{@{}l@{}}
      \\
      τ_7' {=} (2,\, \text{c},\,ℕ_+,\,\mathrm{dec},\,2)
    \end{array}
    &&\begin{array}[t]{@{}r@{{=}(}l@{,\,}l@{,\,}l@{,\,}l@{,\,}l@{)}l}
      τ_4' & 1 & \#       & ℕ     & \mathrm{id} & 2 \\
      τ_8' & 2 & ε        & \{0\} & \mathrm{id} & 3 & \text{,}
    \end{array}
  \end{align*}
  $δ(τ) = 1$ for each $τ ∈ T$, and $δ'(τ') = 1$ for each transition $τ' ∈ T'$.
  \footnote{$L(ℳ) = \{\text{a}^k\#w ∣ k ∈ ℕ, w ∈ \{a, b, c\}^*, \text{a occurs $k$ times in $w$}\}$ and $L(\app{A_\#}{ℳ}) = \{\text{a}^k\#w ∣ k ∈ ℕ, w ∈ \{a, b, c\}^*, \lvert w \rvert ≥ k\}$.}
  
  We use \cref{alg:coarse-to-fine} to obtain the 1-best run of $w = \text{a}\#\text{b}\text{a}$:
  On line~4, we get \(θ' = τ_1'τ_2'τ_4'τ_7'τ_5'τ_8'\) (the only run of $\app{A_\#}{ℳ}$ on $w$).
  Then there are only two possible values for $θ$ on line~7, namely \(θ_1 = τ_1τ_2τ_4τ_7τ_5τ_8\) and \(θ_2 = τ_1τ_3τ_4τ_7τ_5τ_8\) of which only $θ_2$ is a run of $ℳ$, hence the algorithm returns $\{θ_2\}$.  
\end{example}

\paragraph{Outlook.}
The author intends to extend \cref{alg:coarse-to-fine} to use multiple levels of approximation (i.e. multiple approximation strategies that can be applied in sequence) and to investigate the viability of this extension for parsing multiple context-free languages in the context of natural languages.

\section*{Acknowledgements}

The author thanks Mark-Jan Nederhof for fruitful discussions and the anonymous reviewers of a previous version of this paper for their helpful comments.
In particular, \cref{ex:pd-with-pop-star} is due to a reviewer's comment and \cref{ex:implementing-nd-storage-no-powerset-proper-containment} is due to Mark-Jan Nederhof.

\bibliographystyle{eptcs}
\bibliography{references}
\inputencoding{utf8}

\end{document}